\documentclass[10pt]{article}
\usepackage[margin=1in]{geometry}
\usepackage{amsmath,amssymb,graphicx}
\usepackage{amsthm}
\usepackage{url}
\usepackage[ruled,vlined]{algorithm2e}
\usepackage{hyperref}
\usepackage{authblk}
\usepackage[numbers,sort&compress]{natbib}
\usepackage{yfonts}
\usepackage{titlesec}
\titlespacing*{\paragraph}{0pt}{0pt}{5pt}

\newcounter{appsection}

\newcommand{\appsection}[1]{%
  \refstepcounter{appsection}%
  \subsection*{Appendix A\arabic{appsection}. #1}%
}

\usepackage{amsthm}
\usepackage{thmtools}

\newtheoremstyle{tight} % name
  {2pt}    % Space above
  {2pt}    % Space below
  {\itshape}  % Body font
  {}       % Indent amount
  {\bfseries} % Theorem head font
  {.}      % Punctuation after theorem head
  {0.5em}  % Space after theorem head
  {}       % Theorem head spec (can be left empty)

\theoremstyle{tight}

\newtheorem{theorem}{Theorem}
\newtheorem{lemma}{Lemma}[appsection]
\newtheorem{proposition}[lemma]{Proposition}

\newtheorem{definition}{Definition}
\newtheorem{criterion}{Criterion}

\def\reads{\mathfrak{R}}
\def\chromosomes{\mathfrak{C}}
\def\assemblygraph{\mathfrak{G}}
\def\genomepaths{\mathfrak{P}}
\def\multiplicity{\mathfrak{M}}
\def\vertexmultiplicity{\mathfrak{N}}
\def\topology{\mathfrak{T}}
\def\CC{\text{CC}}
\def\CPC{\text{CPC}}
\def\OG{\text{OG}}
\def\SG{\text{SG}}
\def\ROG{\text{ROG}}
\def\DBG{\text{DBG}}
\def\CDBG{\text{CDBG}}
\def\AG{\text{AG}}
\def\SPG{\text{SPG}}
\def\CSPG{\text{CSPG}}

\def\GC{\text{GC}}
\def\OA{\text{OA}}
\def\OGA{\text{OGA}}

\def\Core{\text{Core}}

\def\Labels{\text{Labels}}
\def\Superstrings{\text{Superstrings}}

\def\ExhaustiveMultiplex{\text{Multiplex}^*}

\title{Supregraph: Enabling Information-Optimal Assembly Graph Representation of a Read Set}

\author[1,2,*]{Anton Bankevich}

\affil[1]{Department of Computer Science and Engineering, Pennsylvania State University, University Park, PA, USA}
\affil[2]{Huck Institutes of the Life Sciences, Pennsylvania State University, University Park, PA, USA}
\affil[*]{anton.bankevich@psu.edu}

\begin{document}

\maketitle

\begin{abstract}
The first step in any genome assembly algorithm entails the conversion from the domain of strings and overlaps to the language of graphs and paths, typically using one of the two conventional methods: de Bruijn graphs or overlap graphs. However, both standard approaches are known to have limitations. De Bruijn graphs fail to represent complete information from reads, while the overlap graphs often produce artificial breaks in contigs due to the necessity to discard contained reads as a preliminary step.

In this work we present a mathematical model for genome assembly that provides a formal framework to determine what constitutes a correct conversion of a read set into an assembly graph under the assumption of error-free reads. We prove that a correct representation of a read set exists in the form of a new class of assembly graphs, which we call supregraphs. We show that supregraphs can be constructed by iteratively transforming de Bruijn graphs using the multiplexing procedure, previously employed in the genome assemblers LJA and Verkko. Finally, we demonstrate that, under a set of natural assumptions, supregraphs provide a foundation for constructing theoretically optimal genome assemblies.
\end{abstract}

\newpage

\section*{Introduction}\label{sec:intro}
Genome assembly algorithms have historically been designed to tolerate high sequencing error rates. This was essential in the NGS era, when more than half of $100$ bp Illumina reads in a library could contain errors, and became critical for PacBio CLR and Oxford Nanopore reads, which could exceed $10\%$ error rates. The introduction of PacBio HiFi reads \cite{wenger2019accurate} and duplex nanopore sequencing \cite{silvestre2021pair} marked a major shift, delivering long reads with much higher accuracy (typically fewer than one error per $500$ nucleotides). A single HiFi library is sufficient to produce a human genome assembly more contiguous and accurate than earlier assemblies, which required multiple data types and manual curation \cite{wenger2019accurate}. But does this mean the genome assembly problem is finally solved?

With increased power come more ambitious goals: we are no longer satisfied with assembling a few complete chromosome arms. We now aim for complete, diploid (or even polyploid) telomere-to-telomere assemblies \cite{nurk2022complete, liao2023draft, yoo2025complete}. These ambitious objectives shift the primary challenge of genome assembly away from handling sequencing errors toward accurately reconstructing the most complex and repetitive regions of the genome. Centromeres, sex chromosomes, immunoglobulin loci, and other regions with complex evolutionary histories feature tandem or mosaic repeat structures that remain difficult to assemble, even with nearly error-free sequencing data. Consequently, assemblies of such regions often contain structural flaws \cite{zhu2025closeread}.

In this paper, we set aside the problem of sequencing errors and focus instead on the question that grows increasingly important as read accuracy improves: how to reliably assemble a complex genomic region from error-free data. In practice, assembly algorithms must still account for sequencing errors. However, we argue that the quality of modern reads is high enough that treating genome assembly as an error-free problem provides a sufficiently accurate and useful approximation of real-world conditions.

Intuitively, the error-free genome assembly problem can be stated as follows: given a set of substrings of an unknown string $s$ that together cover $s$, reconstruct $s$. While this formulation captures the essence of the assembly problem, it does not specify what constitutes a correct or optimal reconstruction. By invoking Occam's razor, one might define the simplest (or shortest) possible superstring as the correct answer, leading to the Shortest Superstring Problem (SSP) \cite{kececioglu1995combinatorial}. Although the SSP was initially regarded as a mathematical model of genome assembly, it later became evident that the principle of Occam's razor fails in this context, owing to the highly repetitive nature of genomic sequences \cite{myers1995toward}. Moreover, the SSP formulation fails to capture important aspects of real assembly scenarios, such as the fact that sequencing reads may not contain enough information to reconstruct the complete genome.

Graph-based formulations offer alternative models \cite{medvedev2007computability, shomorony2016information, tomescu2017safe, rahman2022assembler}, where chromosomes are encoded as walks in a graph. The most widely used graphs include de Bruijn graphs \cite{pevzner2001eulerian}, overlap graphs \cite{myers1995toward}, and adaptations such as string graphs \cite{myers2005fragment}, sparse de Bruijn graphs \cite{ye2012exploiting}, and repeat graphs \cite{kolmogorov2019assembly}. However, these graphs do not fully represent the information from a read set: $k$-mer-based graphs discard substantial sequence context, while overlap graphs require discarding of subreads, which can introduce artificial breaks in chromosome paths \cite{nurk2022complete, li2024genome, hui2016overlap, jain2023coverage}. This issue is amplified in diploid genomes with variable read lengths \cite{jain2023coverage}. While this limitation has motivated several practical solutions \cite{kamath2024telomere, cheng2021haplotype}, we argue that the root cause of the problem lies in an incomplete theoretical understanding of the graph models underlying genome assembly. Notably, the same problems arise even under idealized assumptions of error-free, fully covering reads. Motivated by these considerations, we focus on the theoretical foundations of the genome assembly problem. We analyze the information traditionally extracted from a read set by assembly algorithms and propose a natural and precise representation of this information using an assembly graph. The algorithms presented are intended primarily to make the introduced concepts rigorous and are not designed with computational efficiency as a goal. Accordingly, memory- and time-efficient implementations are beyond the scope of this paper and will be described elsewhere.

In Sections~\ref{sec:preliminaries} and~\ref{sec:logic} we present a mathematical framework that formalizes the correctness of representing a read set as a graph. Sections \ref{sec:properties} and \ref{sec:conductors} investigate the properties of graphs that can provide a perfect representation of a read set, leading to the introduction of \emph{supregraphs}, a new type of assembly graph that perfectly preserves information from reads (Section~\ref{sec:supregraph}). In Section~\ref{sec:multiplexing} we further adapt it to reflect the practical genome assembly problem by incorporating constraints on the number of substring occurrences and describe a supregraph construction method based on an iterative process called \emph{multiplexing}. Heuristic variants of the multiplexing algorithm have already been implemented in practice by the LJA \cite{bankevich2022multiplex} and Verkko \cite{rautiainen2023telomere} assemblers to resolve repeats in de Bruijn graphs. Finally, we describe a supregraph-based approach for constructing an optimal (but possibly incomplete) genome assembly under general assumptions (Section~\ref{sec:optimal}). Technical details and proofs of all theorems are provided in the appendix.

\section{Preliminaries}\label{sec:preliminaries}
\paragraph{Simplifying assumptions.}
In this work, we make three simplifying assumptions. (I) All reads are error-free and therefore substrings of the genome string. (II) To avoid edge cases caused by chromosomes ending in repeats \cite{rahman2022assembler}, we treat all chromosomes as cyclic. (III) We ignore the double-stranded nature of DNA and represent sequences and graphs as directed.

(I) is required to construct a deterministic genome assembly model; methods for handling sequencing errors are beyond the scope of this paper. (II) and (III) are introduced to simplify formulations. However, all definitions and results extend naturally to bidirected strings and graphs (Appendix~\ref{app:bidirected}). Appendix~\ref{app:linear-challenges} discusses how the presented theory can be applied to genomes with linear chromosomes.

\paragraph{Strings, string-sets, and chains.}
Most objects in this paper are strings over an alphabet $\Sigma$, which we take to be any finite set of symbols, not just the nucleotide alphabet $\{A, C, G, T\}$. We also consider cyclic strings, representing circular chromosomes, denoted $\langle s\rangle$ (e.g., $\langle ACCAT\rangle$). Standard notions such as position and substring extend naturally to this setting; formal definitions accounting for cyclic strings are provided in Appendix~\ref{app:cyclic-strings}. Let $s$ be a string that has a string $t$ as both a prefix and a suffix, so that $s = tx = yt$. We define the \emph{circularization} of $s$ through the overlap $t$ to be the circular string $\langle x \rangle$. Observe that $\langle x \rangle = \langle y \rangle$, since the strings $x$ and $y$ are cyclic shifts of one another.

Many objects we consider (read libraries, genomes, assemblies) are collections of strings, which we call \emph{string-sets}. A string-set is cyclic if it contains only cyclic strings. A string $s$ is a substring of a string-set $S$ if $s$ is a substring of at least one string of $S$. A string-set $S$ is said to be substring-free if no string in it is a substring of another.

In addition to substrings, we introduce the notion of a \emph{proper infix}. A string $s$ is a proper infix of a string $t$ if $t=asb$ for some non-empty strings $a$ and $b$. A string-set is \emph{proper infix-free} if no element is a proper infix of another. In contrast to a substring-free set, a proper infix-free set permits prefix and suffix relations between elements.

A \emph{chain} is a sequence of strings $r_1,\ldots,r_n$ with overlaps $t_1,\ldots,t_{n-1}$, where $r_i$ and $r_{i+1}$ overlap by $t_i$ characters. A chain is cyclic if $r_1 = r_n$, and a $k$-chain if all overlaps have length at least $k$. The \emph{label} of a (cyclic) chain, $Label(C)$, is the (cyclic) string formed by merging the strings according to their overlaps. We will also refer to $Label(C)$ as a \emph{combined sequence} (in contrast to a concatenated sequence) of $r_1,\ldots,r_n$ when the overlap values are clear from context.

When all strings in a chain are elements of a \emph{read set} (a finite, non-cyclic string-set, given as an input for genome assembly), we call them read chains. Fig. \ref{fig:assembly-graphs}A illustrates a read set and a cyclic $3$-chain. The chain uses each read in the set exactly once but in general reads may be reused multiple times in a read chain.

\paragraph{Assembly graphs.} We introduce a general definition for assembly graphs that encompasses all traditional assembly graph types.

\begin{definition} An \emph{abstract assembly graph} is a directed graph where each vertex and edge is labeled with a string over a given alphabet. For each edge $e$ from vertex $u$ to vertex $v$, $u$ is a prefix of $e$ and $v$ is a suffix of $e$.
\end{definition}
For simplicity, we will identify edges and vertices with their labels. E.g., we say ``$u$ is a prefix of $e$'' rather than ``the label of $u$ is a prefix of the label of $e$''. In Appendix~\ref{app:cyclic-labels} we discuss how we interpret the assembly graph definition to take cyclic strings into account. We intentionally allow an edge to go from a vertex to its suffix (\emph{suffix edge}) or from a prefix of a string to its full version (\emph{prefix edge}). Although prefix and suffix edges do not occur in any traditional assembly graphs, our model requires their introduction.

The definition of abstract assembly graphs is intentionally broad and permits certain degenerate configurations that are not relevant for our purposes. We therefore introduce a more restrictive notion.

\begin{definition}
An \emph{assembly graph} is an abstract assembly graph in which every vertex and every edge lies on at least one directed cycle, and no edge is simultaneously a prefix edge and a suffix edge.
\end{definition}

Every abstract assembly graph can be converted into an assembly graph by first removing all vertices and edges that do not lie on any directed cycle, and then contracting every edge that is simultaneously a prefix edge and a suffix edge.

\paragraph{Walks in assembly graphs.} Walks and circuits in an assembly graph can be naturally associated with chains and cyclic chains, consisting of edge labels with overlaps defined by vertex labels. The label of a walk or circuit $P$, denoted as $Label(P)$, is the label of the corresponding chain. Vertices and edges are allowed to occur multiple times in a circuit, but we only consider circuits that are primitive (i.e., not formed by repeating a shorter circuit). For walks $P$ and $Q$ such that the end of $P$ coincides with the start of $Q$, we denote their concatenation by $PQ$. We say that a walk $P$ is an \emph{inner subwalk} of a walk $Q$ if $Q = APB$ for some non-empty walks $A$ and $B$.

A key property traditionally attributed to assembly graphs is that the true genome corresponds to a path in the graph, called a \emph{genome path}. We adopt this notion, while emphasizing that under the assumption of cyclic chromosomes, a genome “path” is in fact a circuit or a collection of circuits. Nevertheless, we refer to it as a genome path regardless of its actual structure. We distinguish between a genome path and a \emph{chromosome path}, the latter denoting a cycle that corresponds to a single chromosome of the genome.

\begin{figure}[h]
\centering
\includegraphics[width=1\textwidth]{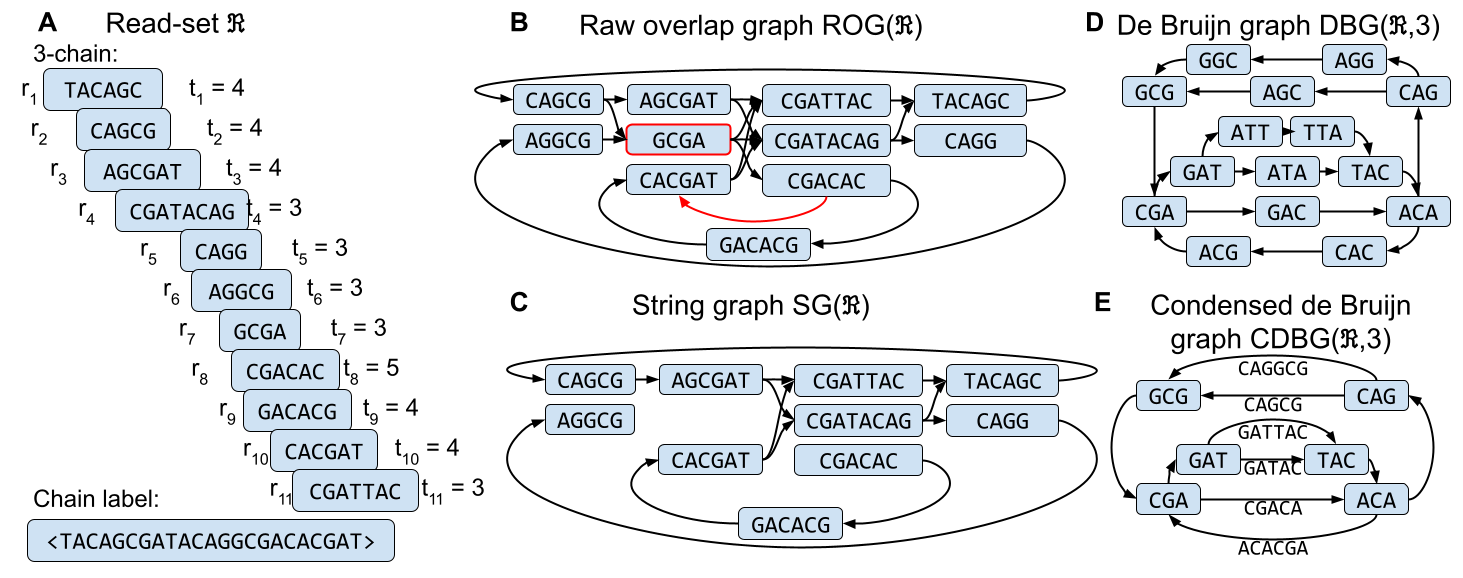}
\caption{\label{fig:assembly-graphs}
\textbf{Assembly graphs}. \textbf{(A)} An example of a cyclic read chain and its label. \textbf{(B)} Raw overlap graph, constructed from all reads including a contained read GCGA highlighted in red. \textbf{(C)} String graph obtained from raw overlap graph by removing one contained read and one transitive edge, highlighted in red. \textbf{(D)} De Bruijn graph. \textbf{(E)} Condensed de Bruijn graph, obtained from de Bruijn graph by transforming each unbranching path into a single edge.}
\end{figure}

\paragraph{Traditional assembly graph types.}
We briefly introduce traditional types of assembly graphs.

In \textit{overlap graphs} vertices are labeled by \emph{non-contained} reads (reads that are not substrings of any other reads), and edges represent overlaps between them. Although edges are not traditionally labeled by strings, they can be assigned the combined sequence of their endpoints joined by the overlap. This represents overlap graphs as an abstract assembly graph, denoted $\OG(\reads)$. We also introduce \emph{Raw Overlap Graph}, $\ROG(\reads)$, that uses all reads, including the contained ones (Fig. \ref{fig:assembly-graphs}B). A \emph{string graph}, $\SG(\reads)$, is obtained from $\OG(\reads)$ by removing transitive edges, whose sequences duplicate labels of longer paths sharing the start and end vertex with the edge (Fig. \ref{fig:assembly-graphs}C).

In \textit{de Bruijn graphs} ($\DBG(\reads,k)$) vertices are $k$-mers and edges are $(k+1)$-mers connecting their $k$-mer prefixes and suffixes (Fig. \ref{fig:assembly-graphs}D). Condensing each non-branching path into a single edge yields the \textit{Condensed de Bruijn Graph}, denoted $\CDBG(\reads, k)$ (Fig. \ref{fig:assembly-graphs}E). Both $\DBG(\reads,k)$ and $CDBG(\reads, k)$ have vertex and edge labels that satisfy our assembly graph definition.

\section{Logic of graph-based genome assembly}\label{sec:logic}
\paragraph{Genome candidates.} A read set may not contain enough information to uniquely reconstruct the genome. Thus the correct genome sequence cannot be assumed to always be the output of the genome assembly problem. A common practical solution is to output contigs, which are the longest substrings of the genome confidently inferred from reads. However, contigs are an artificial construct introduced for convenience in downstream analysis, and would not naturally arise in a model aiming to reflect the process of sequencing. Establishing a direct relationship between the input (reads) and the output (contigs) is therefore difficult. To address this, we introduce the notion of a \emph{genome candidate}: formally, a genome candidate is any collection of cyclic strings. Intuitively  we expect that a genome candidate could have generated the given read set through the sequencing process. This allows us to formally separate the assembly task into two steps: (i) characterize all genome candidates consistent with the reads, and (ii) derive contigs as the longest substrings shared by all genome candidates. Formalization of step (ii) is straightforward:

\begin{definition} For a set of genome candidates $\chromosomes$, \emph{optimal assembly} is a substring-free string-set $\OA(\chromosomes)$, which consists of maximal strings that are substrings of every genome candidate (\emph{contigs}).
\end{definition}

Formalizing step (i) requires explicitly stating all conditions that we expect a read set to impose on the properties of a candidate genome.

\paragraph{Deriving assumptions on genome candidates from conventional logic of genome assembly tools.}
An overlap graph is often considered the most faithful representation of a read set because it allows a natural mapping of read chains, collectively containing all reads as substrings, to paths that visit every vertex. In other words, it represents \emph{genome candidates} ('cyclic chains containing all reads') as \emph{genome path candidates} ('circuits visiting every vertex'). We decompose this definition of genome candidates into three simpler assumptions:

\textbf{Chain condition}. Every chromosome in a genome candidate is a label of a cyclic read chain.

Typically, consecutive reads are expected to overlap by at least $k$ bases, where $k$ is a constant chosen based on sequencing depth and read length. To simplify notation, we absorb the parameter $k$ into the alphabet by working over the alphabet of all $k$-mers. An overlap of length at least $1$ in the $k$-mer alphabet corresponds to an overlap of length at least $k$ in the original nucleotide alphabet. For clarity, however, examples are presented in nucleotide form.

\textbf{Topology condition ($\topology$)}. Genome has exactly one cyclic chromosome.

This is just an example of a possible topology constraint that could be used to characterize the structure of genome candidates. Biological context or experiments (such as karyotyping) can provide additional information about the genome organization, including the number and form of chromosomes and their approximate lengths. We will denote all topology information by $\topology$.

\textbf{Containment condition}. Every non-contained read is a substring of the genome candidate.

It suffices to ensure that non-contained reads are substrings of the genome, since all other reads are contained within them and are therefore also genome substrings. This condition is closely related to another type of information that can be used to characterize genome candidates:

\textbf{Multiplicity condition ($\multiplicity$)}. For each string $s$ from a given collection of sequences, the number of occurrences of $s$ in the genome can assume a limited set of values, defined by $\multiplicity(s)$.

Containment is the simplest example of a multiplicity condition, which requires that the multiplicity of each read sequence be at least $1$. However, sequencing reads encode additional information that can also be expressed naturally in the form of multiplicity conditions. Under the assumption of nearly uniform coverage, the multiplicities of substrings can be estimated from read coverage, although such estimates are inherently probabilistic. Because our model is deterministic, we do not attempt to model this estimation process. Instead, we treat multiplicity information as given, via a function $\multiplicity$ that assigns sequence $s$ with the set of allowed occurrences of $s$ in the genome. For example, if a sequence $s$ is known to be present exactly once in the genome (i.e., \emph{unique}), absent from the genome, or present in the genome at least once, then $\multiplicity(s)$ is $\{1\}$, $\{0\}$, or $\mathbb{Z}_+$, respectively.

We acknowledge that including multiplicity conditions as part of the assembly input may appear counterintuitive. In practice, conclusions about sequence multiplicity are often derived during the assembly process itself. For instance, a long edge in a de Bruijn graph may be treated as unique because the probability of an exact long repeat in a genome is low. Similarly, in a string graph edge between two reads $r_1$ and $r_2$ may be discarded if overlap between them is too short, which can be interpreted as assigning multiplicity $0$ to the combined sequence of $r_1$ and $r_2$.

In our framework, we assume that all such probabilistic analyses have already been performed. The focus of the model is therefore not on how multiplicity information is inferred, but on what can be formally deduced from the collection of multiplicity constraints once they are given.

\paragraph{Mathematical model of genome assembly.}
We argue that chain, containment, multiplicity, and topology conditions encompass all information types available for single long-read library assembly (see appendix~\ref{app::hybrid} for discussion). Other sequencing data types may impose different types of constraints on genome candidates, but these lie beyond the scope of this paper. We thus formalize the mathematical model of genome assembly as follows: each instance of the genome assembly problem is defined by a read set $\reads$, a set of multiplicity conditions $\multiplicity$, and a set of topology conditions $\topology$. Genome candidate is defined as a collection of read chain labels (chromosome sequences) that satisfies the multiplicity and topology conditions. The set of all genome candidates is denoted as $\GC(\reads|\multiplicity,\topology)$. The optimal set of contigs is then derived as $\OA(\GC(\reads|\multiplicity,\topology))$. Genome assembly thus reduces to computing this value efficiently in time and space despite the potential infinitude of the set $\GC(\reads|\multiplicity,\topology)$. 

\paragraph{Representing genome assembly as a graph problem.}
The chain condition plays a different role compared to multiplicity and topology conditions. It defines the building blocks of a genome candidate (possible chromosome sequences), and thus determines the space of genome candidates as collections of \emph{chromosome candidates}. By contrast, the multiplicity and topology conditions restrict this space. This distinction is mirrored in graph models: chromosome candidates map to circuits, defining the space of \emph{genome path candidates} as sets of circuits, while the containment condition imposes constraints on the presence of vertices in these paths.

We therefore define a graph-based representation of a genome assembly problem instance as a two-step process. First, the assembly graph must be constructed that fully preserves the space of possible chromosome candidates as circuit labels. Second, multiplicity constraints ($\multiplicity$) must be encoded as vertex and edge occurrence restrictions ($\vertexmultiplicity$). To formalize this, we introduce the following definitions. 

\begin{definition} A \emph{Chromosome Candidate} ($\CC$) is a label of a cyclic read-chain, that consists of reads from a read set $\reads$. We denote the set of all cyclic read-chains as $Chains(\reads)$ and the set of their labels as $\CC(\reads)$.
\end{definition}
\begin{definition} A \emph{Chromosome Path Candidate} ($\CPC$) is a label of a circuit in an assembly graph $\assemblygraph$. We denote the set of circuits in an assembly graph as $Circuits(\assemblygraph)$ and the set of their labels as $\CPC(\assemblygraph)$. Given a function $\vertexmultiplicity: V(\assemblygraph)\cup E(\assemblygraph)\to 2^{\mathbb{N}_0}$, which specifies occurrence constraints on vertices and edges, and a set of topology constraints $\topology$, we define $GPC(\assemblygraph | \vertexmultiplicity, \topology)$ as the set of all genome path candidates, composed of circuits in $\assemblygraph$, that satisfy the multiplicity and topology constraints specified by $\vertexmultiplicity$ and $\topology$.
\end{definition}
\begin{definition} We say that an assembly graph $\assemblygraph$ is a \emph{chromosome candidate-preserving representation} of a read set $\reads$ if $\CC(\reads)=\CPC(\assemblygraph)$. We say that $\assemblygraph$, $\vertexmultiplicity$ provide a \emph{genome candidate-preserving representation} of a read set $\reads$ and multiplicity conditions $\multiplicity$ if $GC(\reads|\multiplicity)=Labels(GPC(\assemblygraph|\vertexmultiplicity)$).
\end{definition}
We do not include topology conditions in the last definition since they are the easiest to represent in a graph form: the number and lengths of chromosomes directly translate into the number and lengths of circuits. Therefore, for any  genome candidate-preserving representation $(\assemblygraph,\vertexmultiplicity)$ of $(\reads, \multiplicity)$, for any set of  topology constraints $\topology$ we have $GC(\reads|\multiplicity, \topology)=GPC(\assemblygraph|\vertexmultiplicity, \topology)$.

\begin{figure}[h]
\centering
\includegraphics[width=0.9\textwidth]{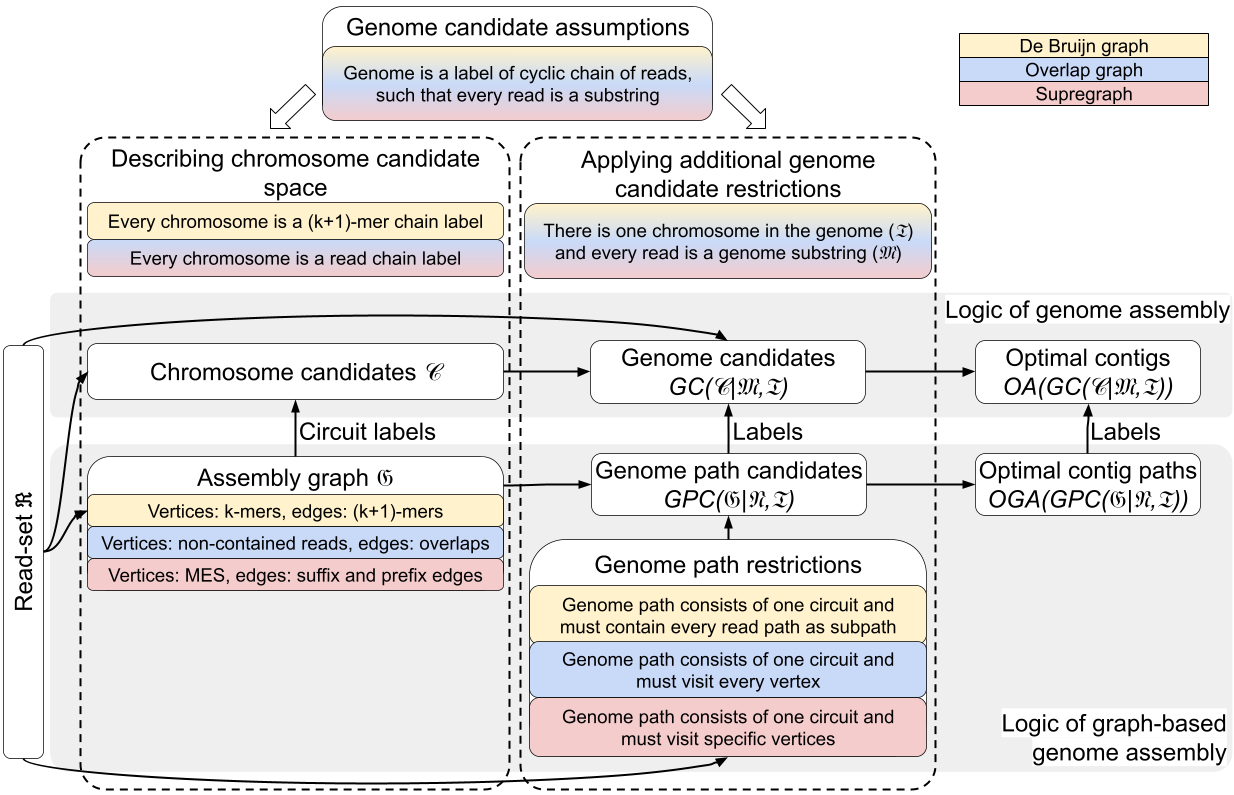}
\caption{\textbf{Logic of graph-based genome assembly}. The general assumptions about genome candidates are divided into two parts: a description of the chromosome-candidate space and additional constraints based on genome topology ($\topology$) and estimated sequence multiplicities ($\multiplicity$). This division is reflected in graph-based assembly by chromosome candidates being represented as circuits in the graph, while the remaining constraints are translated into bounds on vertex occurrences in the genome path ($\vertexmultiplicity$) and restrictions on its topology ($\topology$). Optimal contig paths can then be inferred from the genome-candidate set as the longest walks shared by all genome-path candidates. Variations of details specific to de Bruijn graphs, overlap graphs, and supregraphs are highlighted in yellow, blue, and red correspondingly.}
\label{fig:logic}
\end{figure}

\paragraph{Contig preservation.} The final step of graph-based genome assembly is to extract optimal assembly from graph-based representation of genome candidates. This part can also be done in terms of the assembly graph.

\begin{definition}
    For a collection of genome path candidates $\genomepaths$, \emph{optimal graph-based assembly} is the subwalk-free set of walks $\OGA(\genomepaths)$, consisting of all maximal paths that are subpaths of every genome path candidate.
\end{definition}

Instead of looking for optimal contigs we can look for optimal contig paths, which is a generally simpler task. This concludes the description of graph-based approach to solving the optimal genome assembly problem formulated above. Figure \ref{fig:logic} illustrates how the natural assumptions linking the read set to genome candidates translate into a graph-based pipeline for genome assembly. Practical graph-based genome-assembly tools do not always follow this pipeline explicitly, but the rationale behind their decisions is usually grounded in this underlying logic.

In Figure \ref{fig:logic} the top path from reads to contigs through chromosome and genome candidates represents the mathematical model of genome assembly we introduced. The bottom path from reads to contigs through the assembly graph and genome path candidates represents graph-based approach to genome assembly. For the graph-based approach to function correctly, this diagram must be commutative: the composition of functions along any directed path with the same start and end points must yield the same result.

Two of the three conditions that ensure commutativity ($\CC(\reads)=\CPC(\assemblygraph)$ and $GC(\reads| \multiplicity, \topology)=GPC(\assemblygraph|\vertexmultiplicity, \topology)$) have already been introduced as the chromosome candidate preservation and genome candidate preservation properties, respectively. The remaining relation derived from the diagram is $\OA \circ \Labels = \Labels\circ\OGA$, meaning that for any set of genome path candidates $\genomepaths$, $\OA(\Labels(\genomepaths))=\Labels(\OGA(\genomepaths))$. We refer to this property as \emph{contig preservation}.

These three conditions form the backbone of a formally correct genome assembly and provide us with a problem statement: does a chromosome candidate, genome candidate, and contig-preserving representation exist for each genome assembly problem instance? In the remainder of the paper, we discuss the implications and limitations of these conditions, explain why existing assembly graphs fail to satisfy them, and show how to construct a new type of assembly graph that satisfies them as precisely as possible.

\section{Contig preservation implies properties of assembly graphs}
\label{sec:properties}

\paragraph{Weak contig preservation.}
The contig-preservation condition is surprisingly difficult to satisfy. Assembly graphs direct attention to path labels, whereas optimal contigs may be shorter than even individual vertex labels, preventing the identification of short optimal contigs. We thus restrict our focus to a weaker notion:

\begin{definition}
    An assembly graph satisfies the \emph{weak contig preservation condition} if for any set of genome path candidates $\genomepaths$, $\OA(\Labels(\genomepaths)) \subseteq \Labels(\OGA(\genomepaths)).$
\end{definition}

In other words, in a weakly contig-preserving graph, the labels of optimal contig paths are optimal contigs, although not all optimal contigs may be found in this way.

\paragraph{Conditions for weak contig preservation.}
We introduce two conditions related to weak contig preservation: one necessary and one sufficient. A simple exact criterion of weak contig preservation is difficult to obtain, as some pathological graph structures satisfy it (see appendix~\ref{app:non-redundant}).

\begin{definition}
    A graph is \emph{splitting} if labels of edges that start at the same vertex $v$ (or end at the same vertex $u$) diverge immediately after $v$ (or immediately before $u$).
\end{definition}

By assembly graph definition, labels of all edges outgoing from $v$ share $Label(v)$ as a prefix. In splitting graphs this restriction is stronger: the outgoing edge labels must diverge immediately after that shared prefix. De~Bruijn graphs are splitting; overlap and string graphs are not.

\begin{definition}
    An assembly graph $\assemblygraph$ is \emph{non-redundant} if whenever $P$ and $Q$ are walks in $\assemblygraph$ such that $Label(P)$ is a proper infix of $Label(Q)$, the walk $P$ occurs as an inner subwalk of $Q$.
\end{definition}

This condition prevents the same sequence from labeling distinct walks. Although it could be stated in terms of substrings and subwalks for conventional assembly graphs, the presence of prefix and suffix edges necessitates a formulation using proper infixes and inner subwalks instead (see Appendix~\ref{app:non-redundant}). Non-redundancy is also critical for relating vertex occurrences in a circuit to substring occurrences in its label:

\begin{theorem}\label{thm:nonredcount}
    In a non-redundant graph, for any circuit $C$ and vertex $v$, the number of occurrences of $v$ in $C$ equals the number of occurrences of $Label(v)$ in $Label(C)$.
\end{theorem}

De~Bruijn graphs, condensed de~Bruijn graphs, and string graphs are non-redundant. Raw overlap graphs are typically redundant because subreads and transitive edges cause the same sequence to label different paths.

The following theorem links these structural properties with weak contig preservation:

\begin{theorem}\label{thm:cprescond}
    Every weakly contig-preserving assembly graph is splitting. Every splitting, non-redundant graph is weakly contig-preserving. 
\end{theorem}

Thus, splitting is necessary for weak contig preservation, and splitting combined with non-redundancy is sufficient. This is not an exact characterization, since some redundant graphs are nevertheless weakly contig-preserving. However, we choose to focus on non-redundant graphs because Theorem~\ref{thm:nonredcount} allows multiplicity conditions to be translated into graph constraints and thereby ensures genome path preservation (Section~\ref{sec:multiplexing}).

\paragraph{Does an assembly graph have to be weakly contig preserving?} String graphs are not splitting, and the labels of optimal contig paths in string graphs require additional post-processing to be converted into optimal contigs. Thus, in practice, it is possible to compensate for the lack of the splitting property. However, we will show that such post-processing is unnecessary, since any non-redundant assembly graph can be transformed into a splitting graph (see section~\ref{sec:supregraph}).

The final theorem of this section shows that inability to reconstruct short contigs that do not contain vertices as substrings is the only limitation of assembly graph-based approach based on weakly contig-preserving graphs.

\begin{definition} For string-sets $S$ and $V$, we define $\Superstrings(S, V)$ as the set of strings from $S$ that contain at least one string from $V$ as a substring.
\end{definition}

\begin{theorem}\label{thm:limitation}
For any non-redundant splitting assembly graph $\assemblygraph$ and a set of genome path candidates $\genomepaths$, $$\Labels(\OGA(\genomepaths))=\Superstrings(\OA(\Labels(\genomepaths)), V(\assemblygraph)).$$.
\end{theorem}

From this theorem, we can conclude that a non-redundant splitting graph becomes contig-preserving if each optimal contig contains at least one vertex label as a substring. Theorem \ref{thm:limitation} is a natural limitation of graph-based approaches.

This section introduced key properties of assembly graphs that substantially reduce the search space for chromosome candidate, genome candidate, and contig-preserving representation of a genome assembly problem instance. In the next section, we begin constructing such a representation, incorporating the insights developed here.

\section{Chromosome candidate preservation and conductor strings.} \label{sec:conductors}
In this section, we address the question of whether cyclic read-chain labels can be perfectly represented as labels of graph circuits: $\CC(\reads)=\CPC(\assemblygraph)$. We begin by examining this property in conventional assembly graph models:

\begin{theorem}\label{thm:conventional}
For read set $\reads$
$$\CPC(\OG(\reads))=\CPC(\SG(\reads))\subseteq\CPC(\ROG(\reads))=\CC(\reads)\subseteq\CPC(\DBG(\reads,k)).$$ Moreover, each of these inclusions is strict for some read sets.
\end{theorem}

Figure~\ref{fig:assembly-graphs} illustrates strict inclusions. The cyclic chain CAGCG \text{-} GCGA \text{-} CGACAC \text{-} CACGAT \text{-} CGATTAC \text{-} TACAGC has label $\langle\text{CAGCGACACGATTA}\rangle$, which does not appear in the string graph in Fig.~\ref{fig:assembly-graphs}C because the read GCGA was removed as contained. The de~Bruijn graph in
Fig.~\ref{fig:assembly-graphs}D includes a circuit labeled with
$\langle\text{GACAGC}\rangle$ that is not the label of any read chain. Consequently, neither the string graph nor the de~Bruijn graph provides a
chromosome-candidate-preserving representation for the read set in Fig.~\ref{fig:assembly-graphs}A. String graphs fail this property specifically due to artificial gaps introduced by contained-read filtering. The raw overlap graph is
chromosome-preserving but remains highly redundant and not splitting.

\paragraph{Existence of a chromosome candidate-preserving representation.}
To examine whether a chromosome candidate-preserving representation exists, we consider the following problem. Let $\chromosomes$ denote an arbitrary cyclic string-set. Does there exist a non-redundant assembly graph $\assemblygraph$ such that $\chromosomes = \CPC(\assemblygraph)$? In this formulation, we temporarily disregard the expectation that the chromosome candidate set is constructed as the labels of read chains and instead treat it as an unconstrained set of cyclic strings.

The answer is trivially positive: one can always construct such a graph by representing each cyclic string from $\chromosomes$ as an isolated loop. The essential question is whether there exists a finite graph with this property, since the set $\chromosomes$ is generally infinite. As shown below, this question has a simple and explicit characterization based on the following string property, which will become fundamental for further discussion.

\begin{definition}
A string $v$ is called a \emph{conductor} for a cyclic string-set $\chromosomes$ if, for any pair of strings $va, vb$, that both start and finish with $v$ ($va=a'v$, $vb=b'v$), $\langle ab \rangle \in \chromosomes$ if and only if both $\langle a\rangle\in\chromosomes$ and $\langle b\rangle\in\chromosomes$.
\end{definition}

Although this definition may appear puzzling for an arbitrary chromosome-set $\chromosomes$, it becomes natural when $\chromosomes$ consists of circuit labels of an assembly graph. Indeed, in a non-redundant graph, any two walks that start and end at the same vertex $v$, with labels $va$ and $vb$, can be concatenated into a single walk labeled $vab$ that also starts and ends at $v$. The corresponding circuits then have cyclic labels $\langle a\rangle$ and $\langle b\rangle$ for the original walks, and $\langle ab\rangle$ for the concatenated one. Hence, if $\chromosomes = \CPC(\assemblygraph)$, every vertex label satisfies the conductor condition. The following theorem formalizes and expands this intuition.

\begin{theorem} \label{thm:conductors}
For any non-redundant assembly graph $\assemblygraph$, the following strings are conductors for $\chromosomes = \CPC(\assemblygraph)$:
(i) any superstring of a conductor;
(ii) any vertex, edge, or walk label in $\assemblygraph$;
(iii) any string that is not a substring of $\chromosomes$;
(iv) any string that is at least as long as the longest edge label in $\assemblygraph$.
\end{theorem}

This theorem provides us with an important necessary condition: if a chromosome-set $\chromosomes$ can be represented in the form of circuit labels of a finite assembly graph, then every sufficiently long string is a conductor with respect to $\chromosomes$. This sufficient condition is also necessary:

\begin{theorem}\label{thm:criterion}
If there exists $K$, such that any string of length $K$ is a conductor with respect to chromosome-set $\chromosomes$, then $\CPC(DBG(\chromosomes, K))=\chromosomes$.
\end{theorem}

Combining theorems \ref{thm:conductors} and \ref{thm:criterion} allows us to formulate the criterion:

\begin{criterion}\label{crit:1}
A chromosome-set can be represented as the set of circuits in a finite non-redundant assembly graph if and only if it is \emph{finitely conducting} (if all sufficiently long strings are conductors with respect to $\chromosomes$).
\end{criterion}

We apply this criterion to a set of read chain labels.

\begin{theorem}\label{thm:readscond}
If $\reads$ is a read set with maximal read length $L$, then every string of length $2L$ and every non-contained read is a conductor with respect to $\CC(\reads)$.
\end{theorem}

Combining theorem \ref{thm:readscond} with the formulated criterion implies that every read set has a chromosome-preserving representation in the form of a de Bruijn graph. It is important to note that this de Bruijn graph is not constructed from reads, but rather from chains of reads and the value of $k$ in this graph matches twice the size of the longest read. This graph is fundamentally different from the traditional de Bruijn graph, where $k$ is much smaller than the read length and $k$-mers are extracted only from reads. Such representation is theoretically valid but the size of this graph can easily be exponential in the number of reads. Our next goal is to identify the minimal graph provides a chromosome-preserving representation of a given chromosome-set.

The set of languages encodable by de Bruijn graphs is a well studied class of subregular languages called Strictly Locally Testable (SL) languages~\cite{mcnaughton1971, kim2000, deluca1980}. We have shown that labels of overlapping chains represent a subclass of SL languages.

\section{Supregraph}\label{sec:supregraph}
\paragraph{Non-redundant assembly graphs are defined by their vertices.}
Based on Section~\ref{sec:properties} conclusions, we choose to search for the minimal graph representing a chromosome-set among non-redundant splitting graphs. Our input for this problem is the set of chromosome candidates $\chromosomes$. We note that conventional assembly graphs are defined by the class of strings that can serve as vertices in the graph (e.g. reads or $k$-mers). We generalize this logic to describe an abstract class of assembly graphs defined by a proper infix-free set of possible vertex labels $V$. The algorithm \ref{alg:ag} provides a universal way for constructing non-redundant assembly graphs:

\begin{algorithm}[H]
\caption{$\mathrm{AG}(\chromosomes, V)$}
\label{alg:ag}
    $\assemblygraph \gets$ empty graph\;
    \ForEach{$C \in \chromosomes$} {
        $\mathcal{O} \gets \{(v, p) \mid v \in V \text{ occurs in } C \text{ at position } p\}$\;
        Sort $\mathcal{O}$ by position $p$ along $C$\;
        \ForEach{consecutive pair $(v_1, p_1), (v_2, p_2)$ in $\mathcal{O}$} {
            Add $v_1$ and $v_2$ to $\assemblygraph$\;
            Add edge $v_1 \rightarrow v_2$ with label $C[p_1 : p_2 + |v_2|]$ to $\assemblygraph$\;
        }
    }
    \Return $\assemblygraph$\;
\end{algorithm}

In algorithm \ref{alg:ag} we consider the cyclic sequence of occurrences of vertex labels from $V$ in chromosome candidates from $\chromosomes$ and turn each such sequence into a circuit. All such circuits are collected together and glued according to vertex and edge labels to form the resulting assembly graph $AG(\chromosomes, V)$.
Different choices of $V$ produce familiar graph families: reads, $k$-mers, minimizer $k$-mers, or junction $k$-mers yield string graphs, de Bruijn graphs, sparse de Bruijn graphs, or compressed de Bruijn graphs, respectively. Although we present $AG(\chromosomes, V)$ as pseudocode, the chromosome candidate set can be potentially infinite. Therefore, this algorithm should be viewed primarily as a mathematical construct, describing the relationship between chromosome candidates and non-redundant assembly graphs that could represent them. The following theorem analyzes the expected outcomes of this procedure.

\begin{theorem}\label{thm:ag}
\begin{enumerate}
 \setlength{\itemsep}{0pt}   % no extra space between items
  \setlength{\parskip}{0pt}   % no paragraph spacing in items
  \setlength{\topsep}{0pt}    % no space above/below list
  \setlength{\parsep}{0pt}    % no extra spacing between paragraphs in item
  \setlength{\partopsep}{0pt} % no extra spacing when starting a new paragraph
  \item For any non-redundant assembly graph $\assemblygraph$, we have
    $AG(\CPC(\assemblygraph), V(\assemblygraph)) = \assemblygraph$.
    \item For any infix-free string-set $v$ and chromosome-set $\chromosomes$ $AG(\chromosomes,V)$ is a non-redundant graph.
    \item Let $\chromosomes$ be a cyclic string-set, and let $V$ be a proper infix-free set of strings such that every string in $\chromosomes$ contains at least one element of $V$ as a substring. Then $\chromosomes \subseteq \CPC(AG(\chromosomes, V))$, and for any other non-redundant assembly graph $\assemblygraph$ satisfying $\chromosomes \subseteq \CPC(\assemblygraph)$ and $V(AG(\chromosomes, V)) \subseteq V(\assemblygraph)$, we have $\chromosomes \subseteq \CPC(AG(\chromosomes, V)) \subseteq \CPC(\assemblygraph)$.
\end{enumerate}
\end{theorem}

The first two statements of theorem~\ref{thm:ag} show that a graph can be represented as $AG(\chromosomes, V)$ for certain $\chromosomes$ and $V$ if and only if it is non-redundant. The third statement shows that $AG(\chromosomes, V)$ provides the most precise representation of the chromosome candidate set among all graphs with the same or larger vertex set.

\paragraph{Multiplicity-extremal substrings.}
Traditional vertex label choices share a limitation: they are defined independently of the chromosome candidate set they aim to represent. In $k$-mer graphs, $k$ is arbitrary; in string graphs, labels depend on a specific read set, even though different read sets can yield the same chromosome candidate set. To address this, we define vertex labels directly from the chromosome-set.

Since an assembly graph reflects genome repeat structure, vertex labels should also capture it. We base our choice on maximal exact matches (MEMs)~\cite{kurtz2004versatile}, commonly used in bioinformatics. Since MEMs usually refer to pairs of substring occurrences rather than substrings themselves, for clarity we introduce a separate concept:

\begin{definition} A string $s$ is a \emph{Multiplicity-Extremal Substring} (\emph{MES}) for a chromosome-set $\chromosomes$ if $s$ is a substring of $\chromosomes$ and any proper superstring of $s$ has fewer occurrences than $s$ in at least one chromosome candidtes from $\chromosomes$.
\end{definition}
For a finite text we could make proper superstrings of $s$ to have fewer occurrences in the whole text. Since $\chromosomes$ is a potentially infinite set, we must be careful and specify that the number of occurrences is reduced in at least one chromosome candidate. The following theorem provides a criterion for MES, which does not have this issue.

\begin{theorem}\label{thm:mesalt}
A finite string $s$ is MES for a chromosome-set $\chromosomes$ if and only if there exist characters $c_1\neq c_2$ and $c_3\neq c_4$, such that $c_1s$, $c_2s$, $sc_3$, and $sc_4$ are substrings of $\chromosomes$.
\end{theorem}

The following theorem shows that every string can be extended to a MES with the same collection of occurrences in chromosome candidates.

\begin{theorem}\label{thm:mesext}
For any substring $s$ of a chromosome-set $\chromosomes$, there exists a unique MES, denoted $MES(s)$ (\emph{MES extension} of $s$), such that $s$ occurs exactly once in $MES(s)$ as a substring and every occurrence of $s$ in $\chromosomes$ extends to an occurrence of $MES(s)$.
\end{theorem}

Strings with the same MES extension form an equivalency class of strings that always occur together in any chromosome or genome candidate.

Figure \ref{fig:mes}A illustrates the notion of MES for the chromosome candidate set generated as read chains from the read set shown in Fig. \ref{fig:assembly-graphs}A. The label of a cyclic chain that contains all reads is written on the matrix diagonal. Potentially this matrix is infinite since the label is cyclic and we show only a part of it. Every cell corresponds to a substring of the genome, for which start and end positions are defined by the cell's coordinates. We draw a line to the right or on top of a cell corresponding to a substring that cannot be extended in that direction without reducing its number of occurrences in $\CC(\reads)$. Thus, every cell with lines both to the right and on top corresponds to a MES; all such cells are marked with either an asterisk or letter V. Lines split the diagram into parts with exactly one MES in each at the upper right corner. Substrings from the same component always occur together in a chromosome. Cells that correspond to reads are marked with $\reads$. Properties of MESs are discussed in Appendix~\ref{app:mes}.

\begin{figure}[h]
\centering
\includegraphics[width=0.5\textwidth]{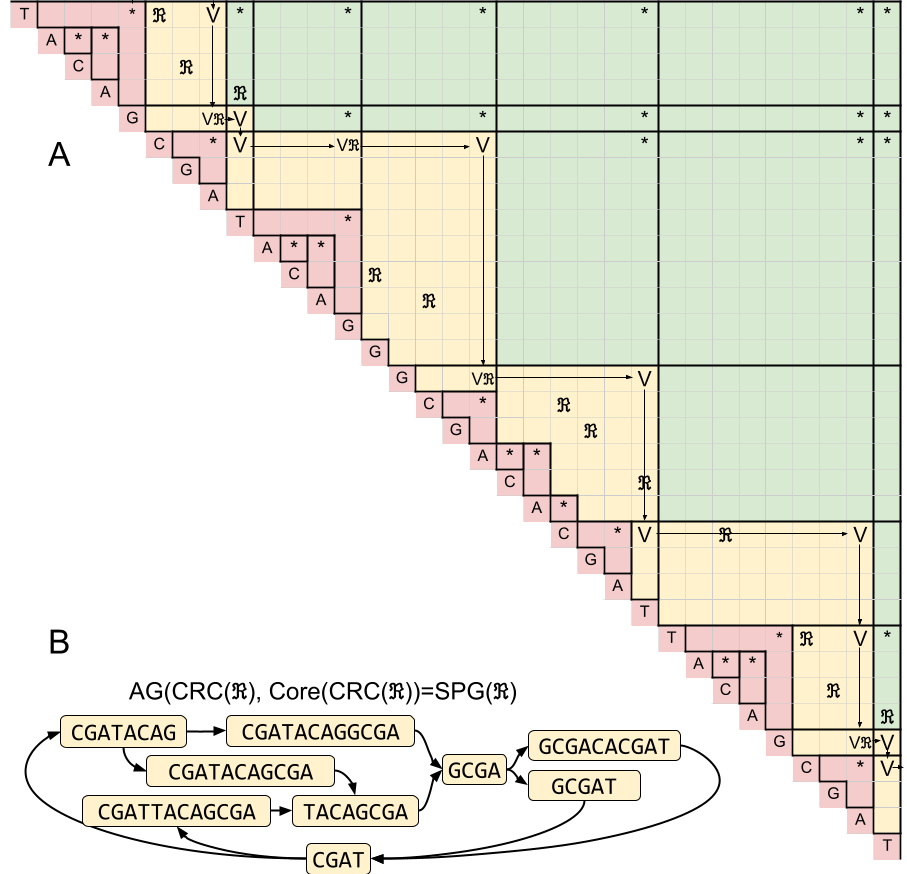}
\caption{\textbf{Multiplicity-extremal substrings}. \textbf{(A)} A table in which every cell corresponds to a substring of a string written on the diagonal. MES are marked with *, nucleotides or V. Non-conductor strings are highlighted in red. Strings that contain core strings as proper infixes are highlighted in green. Core strings are marked with V. Reads are marked with $\reads$. The border between two substrings is drawn if they have different numbers of occurrences in chromosome candidates.}
\label{fig:mes}
\end{figure}

\paragraph{Supregraphs.}
We aim to use MESs as vertex sequences to construct a new type of assembly graph, but a few adjustments are necessary. First, by Theorem \ref{thm:conductors}, only conductor strings can serve as vertices in a non-redundant graph that perfectly represents a set of chromosome candidates; thus, we restrict vertices to strings that are both MESs and conductors. Second, in a non-redundant graph vertices can not be proper infixes of one another; thus, we ignore conductor MES that contain other conductor MES as proper infixes (an approach opposite to subread removal in overlap graph construction).

\begin{definition} A string $s$ is called a \emph{core string} for a chromosome-set $\chromosomes$ if $s$ is a conductor MES, such that no other conductor MES $t$ is a proper infix of $s$. The set of all core strings will be denoted as $\Core(\chromosomes)$.
\end{definition}

Core strings (that correspond to vertices of the supregraph) are marked with letter V in Fig. \ref{fig:mes}A.

\begin{theorem}\label{thm:spg}
Let $\chromosomes$ be a finitely-conducting chromosome-set. Then $\CPC(AG(\chromosomes, \Core(\chromosomes)))=\chromosomes$. Each edge in $AG(\chromosomes, \Core(\chromosomes))$ is either a prefix or a suffix edges.
\end{theorem}

 Theorem \ref{thm:spg} shows that our choice of vertices resulted in a chromosome candidate preserving assembly graph. From the perspective of traditional genome assembly, the resulting graph may appear unnatural (Fig. \ref{fig:mes}B), since prefix and suffix edges are disallowed in conventional assembly graphs. This motivates the introduction of a new type of assembly graph:

\begin{definition} A \emph{supregraph} (Suffix-Prefix Graph) is a non-redundant assembly graph in which every edge is either a suffix or prefix edge. A supregraph $\assemblygraph$ is \emph{condensed} if all labels of its vertices are MESs with respect to $CPC(\assemblygraph)$.
\end{definition}

By theorem~\ref{thm:spg} for any finitely-conducting chromosome-set $\chromosomes$, $\assemblygraph=AG(\chromosomes,\Core(\chromosomes))$ is a supregraph. Moreover, by construction, vertices of $\assemblygraph$ are MES with respect to $\chromosomes=CPC(\assemblygraph)$, showing that $\assemblygraph$ is also a condensed supregraph.

Although this is not immediately apparent from the definition, condensed supregraphs are analogous to condensed de~Bruijn graphs in that they do not contain unnecessary unbranching paths. The structure and properties of condensed supregraphs are discussed in appendix~\ref{app:supregraphs}.

By theorem~\ref{thm:spg}, $AG(\chromosomes, \Core(\chromosomes))$ only has suffix and prefix edges making it a supregraph. By construction all its vertices are MES with respect to $\chromosomes=AG(\chromosomes, \Core(\chromosomes))$, making it a condensed supregraph that perfectly represents the chromosome-set $\chromosomes$. The next theorem shows that it is the smallest such supregraph.

\begin{theorem}\label{thm:minspg}
For a finitely-conducting chromosome-set $\chromosomes$ $AG(\chromosomes, \Core(\chromosomes))$ is the smallest supregraph $\assemblygraph$ (in terms of number of vertices or edges) such that $\CPC(\assemblygraph)=\chromosomes$.
\end{theorem}

Note that the constructed supregraph is not the smallest among all non-redundant splitting graphs representing a given chromosome-set. However, the actual smallest such graph has a structure very similar to $AG(\chromosomes, \Core(\chromosomes))$ and can be obtained from it via a simple procedure discussed in appendix~\ref{app:supregraphs}.

This completes the first goal of this work: we have shown that a chromosome candidate-preserving representation of a read set as an assembly graph exists and identified its minimal form, which we will denote as $SPG(\reads)$.

Figure~\ref{fig:mes} illustrates the construction of a supregraph. Vertex labels are chosen among substrings of reads and read chains, represented as cells in the diagram. We focus on MES, corresponding to the top-right corners of each subarea. Non-conductor strings (red) and conductor MES that contain other conductor MES as proper infixes (green) are ignored. The remaining yellow MES (marked by letter V) constitute the core strings, which are ordered along a chromosome candidate; this order is shown as a path in Figure~\ref{fig:mes}A. Paths from all chromosomes are then merged by gluing vertices and edges with matching labels, producing the final supregraph (Figure~\ref{fig:mes}B).

\paragraph{Transforming any assembly graph into a condensed supregraph.}
Any non-redundant assembly graph $\assemblygraph$ (Fig.\ref{fig:multiplexing}A) can be transformed into a condensed supregraph with the same set of circuit labels using the following \emph{supregraph condensing} procedure: first, each edge $e$ from $u$ to $v$ is replaced with a vertex, labeled $e$, that is connected with $u$ and $v$ through incoming prefix edge and outgoing suffix edge (Fig.\ref{fig:multiplexing}B). Then every edge and vertex sequence is replaced with its MES (Fig.\ref{fig:multiplexing}C). The resulting graph can contain unbranching paths, consisting of vertices, marked with the same label. By collapsing these unbranching paths we obtain a condensed supregraph (Fig.\ref{fig:multiplexing}D), denoted as $\CSPG(\assemblygraph)$. We refer to this procedure as supregraph condensing. Note that since any supregraph is a splitting graph, we have shown that any graph can be transformed into a splitting graph without changing its set of chromosome candidates.

\begin{theorem}\label{thm:cspg}
For any non-redundant assembly graph $\assemblygraph$, $\CSPG(\assemblygraph)$ is a condensed supregraph with the same set of chromosome path candidates: $\CPC(\assemblygraph) = \CPC(\CSPG(\assemblygraph))$.
\end{theorem}

\section{Multiplicity and multiplexing}\label{sec:multiplexing}
\paragraph{Containment condition.}
In Sections \ref{sec:conductors} and \ref{sec:supregraph}, we focused on preserving chromosome candidates. While this was achieved, the resulting graph does not directly link occurrences of non-contained read sequences in genome candidates to vertices or edges in genome path candidates. Here, we show how to restore this connection.

\begin{theorem}\label{thm:mescnt}
For non-redundant assembly graph $\assemblygraph$, vertex $v$, and string $s$, if $Label(v)=MES(s)$ (with respect to $\CPC(\assemblygraph)$), then for any circuit $C$ in $\assemblygraph$, the number of occurrences of $v$ in $C$ is the same as the number of occurrences of $s$ in $Label(C)$.
\end{theorem}

Theorem \ref{thm:mescnt} allows direct translation of containment conditions, but only for reads whose MES extension matches a vertex label. In Fig. \ref{fig:mes}A, most reads fall in the "yellow zone," meaning their MES extensions are vertex labels. No reads fall in the "red zone," as all non-contained reads are conductors by Theorem \ref{thm:readscond}. Read AGCGAT falls in the "green zone" and its MES extension corresponds to the path TACAGCGA \text{-} GCGA \text{-} GCGAT, showing that occurrences of some reads correspond to paths rather than vertices or edges, similar to de Bruijn graphs. To address this, we adopt the multiplexing approach previously used for repeat resolution in de Bruijn graphs \cite{bankevich2022multiplex, rautiainen2023telomere}.

\paragraph{Multiplexing.}
Consider a junction vertex $v$ in a supregraph $\assemblygraph$, defined as a vertex with at least two incoming and two outgoing edges. Any path visiting $v$ can enter via any incoming edge and exit via any outgoing edge. Multiplexing $v$ restricts these possibilities by replacing $v$ with a set of vertices, each representing a specific incoming--outgoing edge pairing (see Fig. \ref{fig:multiplexing}E,G,I and Algorithm \ref{alg:mult}). As multiplexing can introduce unbranching paths in the graph, we apply CSPG procedure to turn the resulting supregraph into condensed supregraph after each step (see Fig. \ref{fig:multiplexing}F,H). We, however, assume that multiplexing procedure connected every incoming edge with at least one outgoing and vice versa to avoid creating vertices that are not present in any cycles.

\begin{algorithm}[H]
\caption{$VertexMultiplex(\assemblygraph, v, D_r)$}
\label{alg:mult}
\ForEach{incoming edge $uv$ in $Graph$}{
    \ForEach{outgoing edge $vw$ in $Graph$}{
        \If{$D_r(uv,vw)$}{
            add vertex $p$ labeled with $Label(uvw)$ to $Graph$\;
            add prefix edge from $up$ to $Graph$\;
            add suffix edge $pw$ to $Graph$\;
        }
    }
}
$Graph \leftarrow CSPG(Graph)$\;
\Return $Graph$\;
\end{algorithm}

\begin{algorithm}[H]
\caption{$\ExhaustiveMultiplex(\assemblygraph, D=(D_c, D_r))$}
\label{alg:fmult}
\While{there exists a vertex $v$ in $Graph$, such that $D_c(v)$}{
    $Graph \leftarrow$ \textsf{VertexMultiplex}($Graph, v, D_r$)\;
}
\Return $Graph$\;
\end{algorithm}

To apply multiplexing iteratively (Algorithm \ref{alg:fmult}) one must define two decision rules: \emph{rule of choice} $D_c$ defines whether a vertex should be multiplexed, and \emph{rule of restriction} $D_r$, that defines which incoming--outgoing edge pairs should be preserved as new vertices. We will refer to a pair of rules, defining iterative multiplexing as \emph{multiplexing decision rule} $D=(D_c,D_r)$. In general, an additional rule that defines the order of vertex processing is required, but we will consider only decision rules such that iterative application of multiplexing yields the same result regardless of the order in which vertices are processed.

Different rules produce graphs with different properties, i.e., if the restriction rule always preserves all edge pairings, the set of chromosome path candidates (CPC) remains unchanged; such steps are called \emph{free}. The following theorem describes the properties of graphs, constructing using multiplexing.

\begin{figure}
\centering
\includegraphics[width=1\textwidth]{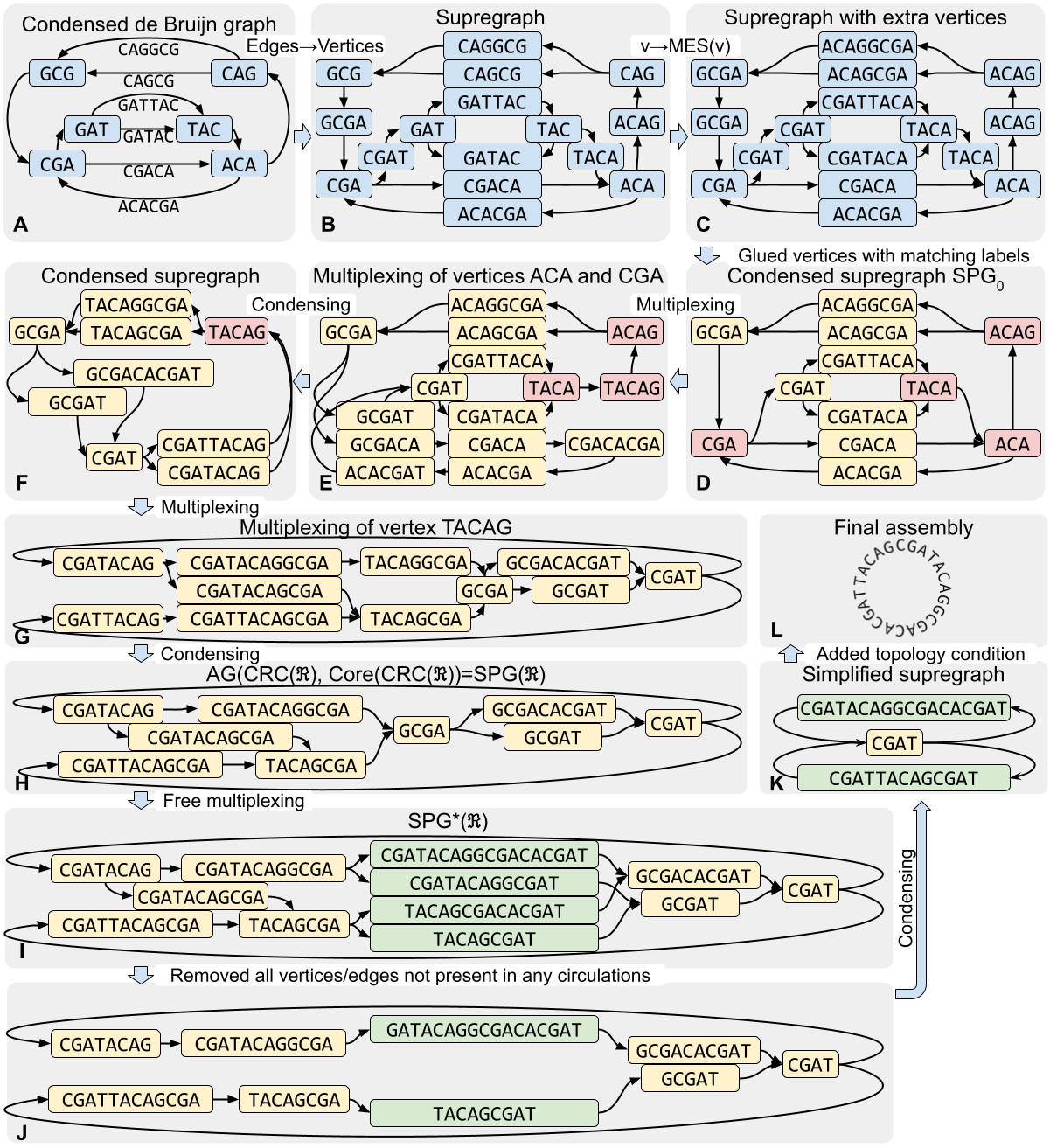}
\caption{\label{fig:multiplexing}
\textbf{Multiplexing}. The figure illustrates a theoretical genome assembly pipeline, including represention a read set in the form of supregraph and its further simplification, yielding an optimal assembly. Yellow/red/green vertices in subfigures D-K represent core vertices, non-conductor vertices and conductor MES that contain other conductor MES as proper infixes, respectively.  \textbf{(A)} Condensed de Bruijn graph, constructed from read set $\reads$. \textbf{(A-D)} supregraph condensing procedure. \textbf{(D-H)} Construction of $\SPG(\reads)$ using multiplexing approach. \textbf{(I)} $\SPG^*(\reads)$: result of additional free multiplexing, that ensures that all MES extensions of non-contained reads are among supregraph vertices. \textbf{(I-K)} OptimalAssembly algorithm: \textbf{(J)} Result of graph cleaning assuming that all non-contained reads are present in the genome exactly once. \textbf{(K)} Result of supregraph cleaning and compressing. \textbf{(L)} Final assembly assuming that the genome is known to consist of a single cyclic chromosome.}
\end{figure}

\begin{theorem}\label{thm:freemultiplexing}
Let $\assemblygraph^*$ be obtained from a supregraph $\assemblygraph$ by multiplexing. Then:
\begin{enumerate}
 \setlength{\itemsep}{0pt}   % no extra space between items
  \setlength{\parskip}{0pt}   % no paragraph spacing in items
  \setlength{\topsep}{0pt}    % no space above/below list
  \setlength{\parsep}{0pt}    % no extra spacing between paragraphs in item
  \setlength{\partopsep}{0pt} % no extra spacing when starting a new paragraph
  \item $\assemblygraph^*$ is a condensed supregraph and $\CPC(\assemblygraph^*) \subseteq \CPC(\assemblygraph)$.
    \item $\CPC(\assemblygraph^*) = \CPC(\assemblygraph)$ if and only if every multiplexing step is free.
    \item Any condensed supregraph $\assemblygraph^*$ can be obtained from $\SPG(\CPC(\assemblygraph^*))$ (minimal supregraph with the same set of chromosome candidates) by a sequence of free multiplexing steps.
\end{enumerate}
\end{theorem}

In other words, the minimal supregraph $\SPG(\chromosomes)$ with a given chromosome candidate set $\chromosomes$ serves as a canonical starting point: any condensed supregraph with the same set of chromosome path candidates can be derived from it via free multiplexing. In contrast, non-free multiplexing strictly reduces the space of chromosome candidates.

\paragraph{Using multiplexing to construct $SPG(\reads)$.} The definition of $SPG(\reads)$ is not constructive as it requires processing of an infinite set $\CC(\reads)$. Instead $SPG(\reads)$ can be constructed from the condensed de Bruijn graph $DBG(\reads, k)$ using multiplexing. To do it we first apply supregraph condensing to $DBG(\reads, k)$ (fig. \ref{fig:multiplexing}A--C), turning de Bruijn graph into a condensed supregraph $CSPG(DBG(\reads, k))$, which we denote as $SPG_0(\reads)$ with the same set of chromosome candidates. Combining theorems \ref{thm:conventional} and \ref{thm:cspg} we have $\CC(\reads)\subseteq\CPC(\DBG(\reads,k))=CPC(SPG_0(\reads))$. Thus to transform $SPG_0(\reads)$ into $SPG(\reads)$ we must first prohibit some extra chromosome path candidates $SPG_0(\reads)$ might have. Multiplexing procedure helps with it, since multiplexing steps can reduce the set of chromosome path candidates. Here we use multiplexing with the following decision rule $D^{SPG}$: multiplex only non-conductor vertices; combine incoming edge $e_1$ and outgoing edge $e_2$ into a new vertex only if combinled label of $e_1$ and $e_2$ is present as a substring in a read chain (in other words it is a substring of $\CC(\reads)$).

\begin{theorem}
    Multiplexing with decision rule $D^{SPG}$ always stops in a finite number of steps and regardless of operation order $\ExhaustiveMultiplex(SPG_0(\reads), D^{SPG})=SPG(\reads)$.
\end{theorem}

Figures \ref{fig:multiplexing}D-G illustrate the application of this algorithm. In Fig. \ref{fig:multiplexing}D, ACA and CGA are non-conducting; testing all edge-pair combinations yields the new vertices shown in Fig. \ref{fig:multiplexing}E. Condensing(Fig. \ref{fig:multiplexing}F) leaves one remaining non-conducting vertex, resolved in Fig. \ref{fig:multiplexing}G. Final condensing (Fig. \ref{fig:multiplexing}H) completes the construction of $SPG(\reads)$.

\paragraph{Using free multiplexing to construct a genome candidate preserving graph.} We get back to the problem of perfectly translating the containment condition into vertex multiplicity condition. By theorem \ref{thm:mescnt}, it is sufficient to construct a chromosome candidate-preserving graph that contains $MES(r)$ as a vertex sequence for every non-contained read $r$ from $\reads$. By theorem \ref{thm:freemultiplexing} all chromosome candidate preserving graphs can be constructed from $SPG(\reads)$ using a series of free multiplexing steps. The following multiplexing decision rule $D^\reads$provides us with a graph that can perfectly represent containment condition in the graph form: we apply multiplexing to a vertex only if it is labeled with a proper infix of a read from $\reads$, all multiplexing steps are free.

\begin{theorem}\label{thm:spgstar}
    Multiplexing with the decision rule $D^{\reads}$ always stops in a finite number of steps and regardless of the operation order the result is the same condensed supregraph $SPG^*(\reads)=\ExhaustiveMultiplex(SPG(\reads), D^{\reads})$. $SPG^*(\reads)$ contains vertex with label $MES(r)$ for every non-contained read $r$ from $\reads$.
\end{theorem}

In Fig. \ref{fig:multiplexing}E, GCGA is such a vertex (a substring of AGCGAT), and multiplexing it (Fig. \ref{fig:multiplexing}I) yields the final graph $SPG^*(\reads)$.

Theorem \ref{thm:spgstar} concludes the central result of this paper. String graph is expected to perfectly translate chain and containment conditions, but fails to do it because of contained read removal. $SPG^*(\reads)$ fulfills these expectations as it perfectly represents chromosome candidates as circuits and enforces containment exactly: a genome path contains a vertex labeled $MES(r)$ if and only if the genome candidate contains read $r$ as a substring, providing a perfect graph translation of both chain and containment conditions.

\paragraph{Multiplicity.}
While we have provided a perfect translation for the containment condition, complete representation of a genome assembly problem in the graph form requires translation of any set string multiplicity conditions ($\multiplicity$) into vertex multiplicity conditions ($\vertexmultiplicity$). Let $S$ be a set of strings with multiplicity constraints defined by function $\multiplicity$. If graph $\SPG(\reads)$ contains vertices labeled $MES(s)$ for every $s\in S$, then the perfect translation of multiplicity condition can be obtained the same way as for containment condition: for each string $s\in S$ there is a vertex  that always has the same number of occurrences in genome path candidate as $s$ has in its label. In particular, if $S$ consists only of non-contained read sequences, we can always perfectly translate the multiplicity condition into the graph form.

This allows us to extend the main result by taking multiplicity and topology conditions into account: given a read set, multiplicity constraints on non-contained reads $\multiplicity$, and topology conditions $\topology$ we have constructed a supregraph $SPG^*(\reads)$ and a set of vertex multiplicity constraints $\vertexmultiplicity$ that provide a genome candidate-preserving representation of genome assembly problem instance.

 % We discuss the possibility of constructing genome candidate-preserving translation into graph form in case multiplicity condition is not limited to read sequences in APPENDIX.

\paragraph{Multiplexing in practice.} While multiplexing provides constructive definitions for $SPG(\reads)$ and $SPG^*(\reads)$, its practical application requires addressing additional challenges. First, construction of $SPG(\reads)$ requires quick query answers of two types: whether a given string is a conductor and whether a given string is a substring of a read chain. We leave these questions of efficient data structures beyond the scope of this paper. Second, the number of multiplexing steps in these procedures is difficult to estimate since it applies to each MES string that has length at least $k$ and is a substring of one of the reads.

We introduced the multiplexing procedure to address a theoretical question, but this approach is also very powerful in practice. Multiplexing allows us to selectively ignore specific repeats while preserving all other matches. For example, de Bruijn graphs are often criticized for merging identical $k$-mers that occur in the middle of reads even when the reads themselves do not overlap (e.g., reads with sequences ARB and CRD). Multiplexing the vertex corresponding to such a match (R) produces a graph in which this spurious middle-repeat match is ignored, while all other structure is retained. Exhaustive free multiplexing within a tangled region of the graph can generate all possible paths through this region and enable verification of their viability using additional data. Finally, even a small number of multiplexing steps applied to graphs constructed from real sequencing reads may expose errors in the form of simple graph structures (bulges or tips), thereby facilitating error correction.

These observations suggest several promising directions for future development of multiplexing-based genome assembly algorithms.

\paragraph{Multiplexing and supregraphs.}
In \cite{bankevich2022multiplex}, the multiplexing procedure was used to transform $DBG(\reads, k)$ into $DBG(\reads, k+1)$. Iterative application of this procedure enabled a smooth transition from smaller to larger values of $k$, avoiding the artificial gaps that are often introduced when de Bruijn graphs are constructed independently for successive values of $k$. \cite{peng2010idba,bankevich2012}. 

However, in \cite{bankevich2022multiplex} it was shown that the resulting graph does not always support all required operations and that certain vertices must be “frozen.” This limitation stems from the prohibition of prefix and suffix edges, which makes de Bruijn graphs insufficiently expressive to smoothly represent all possible sequence configurations.

To illustrate this point, consider a read set $\reads_k$ consisting of all possible $k$-mers. The minimal representation of this read set is the graph $\DBG_k$, whose vertices are all $4^{k-1}$ $(k-1)$-mers and whose edges are all $4^k$ $k$-mers. The containment condition can then be expressed as: every edge must appear at least once along the genome path.

Now suppose we add a single read of length $k+1$ to obtain a new read set $\reads'_k$. How does the assembly graph change? If prefix and suffix edges were allowed, the modification would require only a single multiplexing operation. In contrast, when such edges are prohibited, the graph $\DBG_{k+1}$ (with all $k$-mers as vertices and all $(k+1)$-mers as edges) becomes the minimal non-redundant splitting graph capable of representing $\reads'_k$ (see appendix~\ref{app:multiplexing}). Adding just a single read leads to a change in the entire graph, transforming $\DBG(\reads_k, k-1)$ into $\DBG(\reads_k, k)$. Thus, beyond the natural choice of vertex labels, a supergraph structure is essential for multiplexing.

\section{Optimal genome assembly}\label{sec:optimal}
We conclude with an example of a solution to the optimal genome assembly problem in a simplified setting. We assume that no topology constraints are imposed, and that multiplicity conditions apply only to non-contained reads, with each allowed multiplicity being an interval of integers. Algorithm~\ref{alg:optimal_assembly} describes a procedure that, given a read set $\reads$ and a set of multiplicity constraints $\multiplicity$, identifies all longest strings that appear in at least one genome candidate and contain at least one unique read (that is, a read whose multiplicity is constrained by $\multiplicity$ to be exactly~$1$) as a substring.

\begin{algorithm}[H]
\caption{OptimalAssembly($\reads|\multiplicity$)}
\label{alg:optimal_assembly}
$\assemblygraph \leftarrow SPG^*(\reads)$\;
$\vertexmultiplicity \leftarrow$ vertex multiplicity conditions induced by $\multiplicity$\;
$N \leftarrow $ a network on graph $\assemblygraph$ with (upper and lower) vertex capacities in $N$ defined by $\vertexmultiplicity$\;
\ForEach{$t \in V(\assemblygraph)\cup E(\assemblygraph)$}{
    \If{$t$ does not belong to any feasible circulation in $N$}{
        remove $t$ from $\assemblygraph$\;
    }
}
$\assemblygraph \leftarrow$ \CSPG($\assemblygraph$)\;
$Assembly \leftarrow \varnothing$\;
\ForEach{$r \in \reads$}{
    \If{$\multiplicity(r)=\{1\}$} {
        Add shortest vertex containing $read$ ($MES(r)$) to $Assembly$\;
    }
}
\Return $Assembly$\;
\end{algorithm}

Algorithm~\ref{alg:optimal_assembly} removes all vertices and edges that cannot occur in any genome consistent with the specified multiplicity constraints. The remaining graph is then condensed to collapse maximal unbranching paths. The following theorem shows that this natural supregraph cleaning procedure is sufficient to recover all optimal contigs that contain at least one unique read as a vertex sequence.

\begin{theorem}\label{thm:oa}
For any read set $\reads$ and a set of multiplicity conditions $\multiplicity$ defined on non-contained reads:
$$OptimalAssembly(\reads|\multiplicity)=\Superstrings(OA(GC(\reads|\multiplicity)), \reads^1),$$
where $\reads^1$ denotes the set of unique reads from $\reads$.
\end{theorem}

Figure \ref{fig:multiplexing}J illustrates the effect of removing all vertices not belonging to any circulation, under the assumption that every non-contained read appears exactly once. The remaining condensed graph (Fig.~\ref{fig:multiplexing}K) consists of two loops. Sequences of these loops make up the optimal genome assembly. If an additional topological constraint, requiring the genome to form a single cyclic chromosome, is imposed, the complete genome assembly can be obtained by combining these two loops into a single cycle (Fig.~\ref{fig:multiplexing}L).

\section*{Discussion}
The theoretical framework developed in this paper clarifies the structure of the genome assembly problem and provides a foundation for designing and analyzing practical algorithms. Classical assembly approaches often rely on operations, such as discarding contained reads or resolving repeats, that are justified mainly by intuition or empirical experience. By contrast, the framework of chromosome and genome candidates, supregraphs, and multiplexing specifies exactly which representations preserve information and why. As a result, the entire process becomes more transparent: the representation of the read set, the meaning of every vertex and edge, and the effect of each graph modification step can all be traced directly to well-defined mathematical conditions.

This increased transparency is especially important for multiplexing, which in practice is used as a heuristic repeat-resolution technique. The theoretical formulation shows that multiplexing can be a correctness-governed tool whose effect on the space of genome candidates is well characterized. This not only explains why multiplexing-based strategies succeed in assemblers such as LJA and Verkko, but also provides a framework for evaluating, comparing, and improving the decision rules used in practice. A mathematically grounded representation makes it possible to understand which ambiguities or errors are inherent to the data and which arise from algorithmic choices.

At the same time, several challenges remain before the theory can be translated into practical tools. The ability of our approach to partially reconstruct optimal assembly (which was proven to be an NP-hard problem \cite{kapun2013de}) is explained by the fact that in some cases the supregraph, representing a given set of reads, is exponential in size. Construction of supregraphs requires nontrivial operations, such as determining which substrings represent core MES and verifying whether candidate labels occur in read chains. Moreover, the framework assumes error-free reads, while practical applications require error-correction procedures integrated with graph construction. Addressing these issues and developing efficient data structures for representing and updating supregraphs is beyond the scope of this paper and will be discussed elsewhere.

\section*{Acknowledgements}
The author is grateful to Sergey Nurk, Pavel Pevzner, and Yana Safonova for many helpful discussions that contributed to this work.
\bibliographystyle{unsrt}
\bibliography{references}
\newpage

\appendix
\section*{Appendix}

\appsection{Bidirected strings and graphs}\label{app:bidirected}
While we traditionally represent DNA as a sequence of nucleotides, the real physical DNA molecule contains two chains with nucleotide sequences reverse-complementary to each other. Thus the correct way to encode DNA sequence is with an unordered pair of strings $s_1$, $s_2$ where $s_1$ and $s_2$ are reverse-complements to each other. If DNA is encoded this way, so should reads, chromosome/genome candidates, labels of vertices and edges in assembly graphs. This correction leads to a bidirected assembly graph instead of a conventional directed graph. Bidirected graphs are well-studied and one can easily formulate the same notions of walks, circuits, and their labels in bidirected graphs. In this text for simplicity we do not use the language of bidirected graphs, but every single formulation and proof here can be directly translated into correct formulation and proof for bidirected graphs.

\appsection{Challenges of linear chromosome assembly}\label{app:linear-challenges}
The assumption that all chromosomes in the genome form a single circular sequence is particularly convenient for graph-based models. Without this assumption, each read could, in principle, originate from an independent chromosome, necessitating additional complex constraints to avoid such ``perfect'' assemblies.

In reality, complex genomes predominantly consist of linear chromosomes. However, the practical distinction between circular and linear chromosomes diminishes as chromosome length increases. Complex repetitive regions are interspersed with simple, easily assembled regions, enabling the decomposition of the assembly problem into smaller subproblems, each representing a local ``genome tangle'' with multiple entry points. Conceptually, one can connect all free ends within a subproblem into a single virtual vertex, effectively circularizing each genome fragment involved in the tangle and reducing the assembly problem to the circular chromosome case. This reduction is generally valid, except for subproblems that include actual chromosome termini, which require special treatment and are challenging to model accurately within a graph framework.

\appsection{Cyclic strings and related notions}\label{app:cyclic-strings}
The easiest way to introduce cyclic strings is to consider them as a special case of infinite strings. We say that a string $s$ is a mapping from the index set $(a,b)$ of integers strictly between $a$ and $b$ to alphabet $\Sigma$, where $a,b\in\mathbb{Z}\cup\{\pm\infty\}$. Length of the string $s$ is calculated as $b-a-1$ and denoted as $|s|$. A string $s$ is left or right infinite if $a$ or $b$ is $\pm\infty$  respectively. Otherwise, $s$ is called left or right finite. Two strings $s_1:(a_1,b_1)\to\Sigma$ and $s_2:(a_2,b_2)\to\Sigma$ are equal if and only if there exists an integer number $t$, such that $a_1 = t+a_2$, $b_1 = t+b_2$ and $s_1[i+t]=s_2[i]$ for any integer i from $(a_2,b_2)$. String $s: (a,c)\to\Sigma$ is said to be a concatenation of strings $s_1:(a,b)\to\Sigma$ and $s_2:(b-1,c)\to\Sigma$ if $s[i]=s_1[i]$ for $a<i<b$ and $s[i]=s_2[i]$ for $b-1<i<c$. We omit technical proof that concatenation is properly defined for any pair of right-finite and left-finite strings. $s_1$ and $s_2$ are referred to as prefix and suffix of $s$. String $s$ is a substring of string $t$ if $t = AsB$ for some strings $A$ and $B$. If either $A$ or $B$ are not empty, $s$ is a proper substring of $t$. If both $A$ and $B$ are not empty, $s$ is a proper infix of $t$. A right and left infinite string $s:\mathbb{Z}\to\Sigma$ is called cyclic if there exists $d$, such that for any integer $t$ $s[t]=s[t+d]$. The smallest such $d$ is referred to as period. For a finite string $s$ we denote $\langle s\rangle$ as a cyclic string with period $|s|$, such that $\langle s\rangle[t]=s[t \pmod{|s|}]$ for any integer $t$. Note that we use this notation only if $|s|$ is the period of this string. In other words, $s$ cannot be represented as a shorter string, repeated multiple times.

These definitions are in agreement with the standard definitions for finite strings. However, the definition of position inside a string should be introduced more carefully than just the distance from the start.

\begin{definition} \emph{Position} in a string $s$ is a pair of strings $(l,r)$, such that $s=lr$ and $r$ is not empty. The character at position $(l,r)$ is the first character of $r$. Two positions $l_1r_1$ and $l_2r_2$ in string $s$ are equal iff $l_1=l_2$ and $r_1=r_2$.
\end{definition}
For finite strings position can be associated with the length of $A$. Positions in cyclic strings are more tricky: for a cyclic string $\langle s\rangle$, there are only $|s|$ different left-final suffixes and $|s|$ different right-final prefixes, resulting in $|s|$ different positions. In contrast from finite strings, where positions can be associated with non-negative integer numbers, for cyclic strings positions should be associated with elements of finite group $\mathbb{Z}/|s|\mathbb{Z}$, or, more precisely, with an $|s|$-element set, on which $\mathbb{Z}/|s|\mathbb{Z}$ group acts. Therefore, we can add integer numbers to positions: e.g. for position $(l,cr)$ there is the next position $(l,cr)+1=(lc,r)$ (here $c$ is a single letter).

\begin{definition} For strings $s$ and $t$, we say that $t$ \emph{occurs at position} $p=(l,r)$ in $t$ if $p$ is a position in $t$ and $s$ is a prefix of $r$. We refer to a pair $(p, s)$ as an \emph{occurrence} of $s$ in $t$.
\end{definition}
Note that since the number of different positions in a cyclic string is finite, the number of occurrences of any substring in a cyclic string is also finite.

\begin{definition} If $s$ occurs in $t$ at position $p=(A,sB)$ and $t$ occurs in $w$ at position $q=(C,tD)$, then $p \circ q=(CA,sBD)$ is the \emph{nested occurrence} of $s$ in $w$.
\end{definition}
Note that while we allowed all kinds of infinite strings here, in the graph assembly model we only consider finite and cyclic strings.

Next we formalize the notion of consequtive occurrences. This notion is natural for finite strings, but for cyclic strings it requires more technical definition.
\begin{definition}
Let $(p_1,s_1),(p_2,s_2),\ldots,(p_n,s_n)$ be a sequence of occurrences of finite strings in a cyclic string $\langle t\rangle$. We say that this sequence of occurrences is \emph{sorted} and its elements are \emph{consecutive occurrences} if there exist non-negative integers $0 = d_1 \le d_2 \le \cdots \le d_n < |t|$
such that $p_i = p_1 + d_i$ for all $i$, for $i=1,\ldots n-1$, $d_i + |s_i| \le d_{i+1} + |s_{i+1}|$, and $d_n+|s_n| \le |t| + |s_1|$. 
\end{definition}

In other words in a sorted list of occurrences, both left and right ends of occurrences are sorted.

\begin{proposition}
    For an infix-free string-set $S$, the set of all their occurrences in a cyclic string $\langle t\rangle$ can be recorded as a sorted list of occurrences.
\end{proposition}
\textbf{Proof.} Let $s$ be one of the strings in $S$ of minimal length that has at least one occurrence in $\langle t\rangle$, and let $(p,s)$ be one such occurrence. The set $\{p,p+1,\ldots,p+|t|-1\}$ contains each position of the cyclic string $\langle t\rangle$ exactly once. Hence, every position can be uniquely written as $p+i$ with $0\le i<|t|$. We call $i$ the \emph{relative position}.

We treat $(p,s)$ as the first element of the list and sort all remaining occurrences of strings from $S$ by increasing relative position, breaking ties by increasing substring length. Let the resulting order be $(p_1,s_1)=(p,s),(p_2,s_2),\ldots,(p_n,s_n)$.

By construction, the left endpoints in this list are sorted. We will show that this list is sorted by right endpoints too. Assume, in search of contradiction, that for some $i$ we have $d_i+|s_i|>d_{i+1}+|s_{i+1}|$. We consider two cases.

\emph{Case 1:} $d_i=d_{i+1}$.  
Since ties in relative position are resolved by increasing substring length, we have $|s_i|<|s_{i+1}|$. Therefore, $d_i+|s_i|=d_{i+1}+|s_i|\le d_{i+1}+|s_{i+1}|$, a contradiction.

\emph{Case 2:} $d_i<d_{i+1}$.  
Then $d_i<d_{i+1}<d_{i+1}+|s_{i+1}|<d_i+|s_i|$, which implies that $s_{i+1}$ is a proper infix of $s_i$. This is impossible for strings from $S$.

To show that $d_n+|s_n|\le |t|+|s_1|$ we again note that if opposite is true then $s_1$ must be a proper infix of $s_n$.
The inequality $d_n+|s_n|\le |t|+|s_1|$ is proved analogously, and we conclude that the constructed list is a sorted list of all occurrences of strings from $S$ in $\langle t\rangle$.\qed

\begin{definition}
    For a proper infix-free set of strings $S$ and a string $T$ we define $Occ(S, T)$ as a sequence of occurrences of strings from $S$ in $T$.
\end{definition}

\appsection{Cyclic labels in assembly graphs}\label{app:cyclic-labels}
Our model allows vertices and edges in the graph to have cyclic labels. In this section we discuss how such vertices and edges should be treated in the assembly graph. Consider a vertex with a cyclic label $\langle s\rangle$. If it has at least one outgoing edge, this edge must be marked with sequence t, such that $\langle s\rangle$ is a prefix of t. Infinite string $\langle s\rangle$ can only be a prefix of itself, thus the edge should also be marked with $\langle s\rangle$. But what about the end vertex of the edge, marked with a cyclic string? The end vertex must be labeled with a suffix of $\langle s\rangle$. Since we allow only finite and cyclic strings, the label of the end vertex must be also $\langle s\rangle$, making this edge a loop. Therefore, while cyclic labels are allowed in assembly graphs, they appear only as isolated loops. Each such loop represents a circuit, thus contributing to the set of chromosome path candidates.

\appsection{Single long read library assembly and hybrid assembly problems} \label{app::hybrid}
Modern challenges in genome assembly require integrating multiple types of data to produce the most contiguous assemblies possible. Assembly graphs effectively represent simple, high-quality long-read libraries but struggle to directly incorporate information from more complex data types, such as Hi-C, genome maps, or error-prone long reads. Nevertheless, a reliable graph built from high-quality reads is widely regarded as the best foundation for genome assembly, effectively reducing the space of possible genomes to a managable space with additional technologies used to close gaps and resolve complex tangles in the graph. Therefore, although this study focuses on the simplest type of input data, it still provides a foundational framework for modern genome assembly applications.

\appsection{Contig preservation and non-redundant graphs}\label{app:non-redundant}
In this section, we discuss the properties and alternative formulations of non-redundant assembly graphs. We start by analyzing the clean graph assumption made in section~\ref{sec:properties}
\begin{proposition}
In a graph in which every vertex and every edge belongs to at least one cycle, each weakly connected component is strongly connected.
\end{proposition}

\textbf{Proof.} Let vertices $u$ and $v$ belong to the same weakly connected component. Then there exists a sequence of vertices $W$ starting in $u$ and finishing in $v$, such that consecutive vertices are connected with an edge in at least one of the directions. We will construct a walk from $u$ to $v$. For every pair of consecutive vertices $w_1$, $w_2$ in $W$, if the edge is directed forward, we add it to the walk. If the edge is directed backward we note that this edge must be a part of a cycle and thus there exists a path from $w_1$ to $w_2$, which we can add to the walk. This method allows one to construct directed walks in both directions between any pair of the vertices from the same weakly directed component, thus every weakly connected component is strongly connected. \qed

\paragraph{Graph examples.}
\begin{figure}[h]
\centering
\includegraphics[width=1\textwidth]{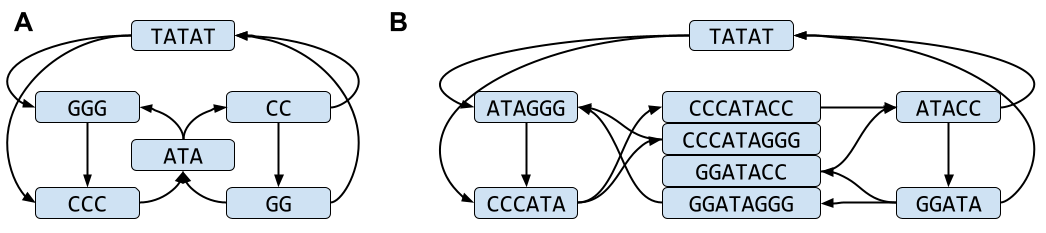}
\caption{\label{fig:cpe}
\textbf{Complex graph examples}. \textbf{(A)} An example of a redundant graph that is weakly contig-preserving and splitting. \textbf{(B)} The result of a non-redundant graph with the same set of chromosome candidates. }
\end{figure}

Figure~\ref{fig:cpe} shows examples of complex cases, where contig preservation is disconnected from redundancy. For clarity in this example all edge labels are ommitted, but each of them is either labeled by concatenation of start and end vertex labels or is a prefix/suffix edge (in subfigure~\ref{fig:cpe}B). The graph in subfigure~\ref{fig:cpe}A is redundant since the vertex label TATAT contains the vertex label ATA as a proper infix. On the other hand, this graph is weakly contig preserving. Such situation occurs because graph structure forces a situation where every genome candidate that contains TATAT, must also contain paths CCC\text{-}ATA and ATA\text{-}CC. As a result, whenever vertex TATAT is in one of the paths from $\OGA(\genomepaths)$ it is impossible for vertex ATA to form a maximal path from $\OGA(\genomepaths)$. Thus even though the same sequence is present multiple times in the assembly graph, it is irrelevant because global graph structure never allows paths from $\OGA(\genomepaths)$ have labels that are substrings of each other.

Non-redundant assembly graph in figure~\ref{fig:cpe}A can be easily transformed into a non-redundant graph with the same set of chromosome candidates by a single free multiplexing step at vertex ATA (Fig.~\ref{fig:cpe}B). The downside of this transformation is that the already weak contig preservation becomes even weaker: consider the set of two genome candidates $\langle GGGCCCATA\rangle$ and $\langle GGCCATA\rangle$. In the graph in Fig.~\ref{fig:cpe}A these genome candidates are represented by the cycles GGG\text{-}CCC\text{-}ATA and GG\text{-}CC\text{-}ATA. The optimal contigs for this set of genome candidates are CC, GG, and ATA. However, only vertex ATA is present in every genome candidate path. The optimal contigs CC and GG are therefore not preserved in the resulting assembly. Resolving vertex ATA leads to the loss of ATA as well.

\paragraph{Non-redundant graph properties.}
The definition of non-redundant assembly graphs is technical and non-intuitive. In simple words its intention is to prevent the same sequence from occurring as a label in multiple places in the graph. However, accurate formulation requires more precise wording. Additional complication is introduced by the presence of prefix and suffix edges. To clarify the meaning of non-redundancy, we provide the following implications of the second condition for the case of a traditional assembly graphs.

\begin{proposition}\label{prop:nrtr}
    Let $\assemblygraph$ be a non-redundant assembly graph without suffix and prefix edges. Then:
    \begin{enumerate}
        \item For any two walks $P$ and $Q$ in $\assemblygraph$, $Label(P)$ is a prefix/suffix/substring/equal of $Label(Q)$ iff $P$ is a prefix walk/suffix walk/subwalk/equal of $Q$.
        \item For any vertex or edge $v$ and walk $P$ in $\assemblygraph$, the number of occurrences of $v$ in $P$ is the same as the number of occurrences of $Label(v)$ in $Label(P)$.
        % \item For any vertex or edge $t$ and circuit $P$ in $\assemblygraph$, the number of occurrences of $t$ in $P$ is the same as the number of occurrences of $Label(t)$ in $Label(P)$.
        % \item For any vertex $t$ and walk $P$ in $\assemblygraph$ and any occurrence of $Label(t)$ at position $p$ in $Label(P)$ there exists a unique prefix path $Q$ of $P$, such that $Q$ ends in $t$ and $|Label(Q)|=p+|Label(t)|$.
    \end{enumerate}
\end{proposition}
\textbf{Proof: (1)} If $Label(P)$ is a proper infix of $Label(Q)$, then by non-redundancy of $\assemblygraph$, $P$ is a subwalk of $Q$.

If $Label(P)$ is a proper suffix of $Label(Q)$, consider a walk $Q'$ obtained from $Q$ by appending one additional edge $e$ to its end (such an edge exists because every vertex lies on a cycle and therefore has at least one outgoing edge). Then $Label(Q)$ is a proper prefix of $Label(Q')$, and consequently $Label(P)$ is a proper infix of $Label(Q')$. By definition, $P$ is an inner subwalk of $Q'$. Since $Q'$ differs from $Q$ only by the addition of the final edge, it follows that $P$ is a subwalk of $Q$.

The cases where $Label(P)$ is a proper prefix of $Label(Q)$ or $Label(P)=Label(Q)$ are handled analogously.

\textbf{(2)} We only prove this statement for vertex $v$. The proof for edges is analogous. 

First we note that considering a walk $Q$ that consists only of vertex $v$ in the definition of a non-redundant graph gives us the following conclusion: for any walk $P$ that does not contain vertex $v$ as an inner vertex, $Label(P)$ does not contain $Label(v)$ as a proper infix. We can split walk $P$ into subwalks by breaking $P$ at each occurrence of $v$. Each subwalk does not contain $v$ as an inner vertex, thus its label does not contain $v$ as a proper infix. Therefore $v$ can only appear as a prefix or suffix of subwalks. From item 1 we know that it can be a prefix/suffix only if corresponding subwalk starts/finishes in $v$. Thus the number of occurrences of $v$ in $Label(P)$ is the same as the number of occurrences of vertex $v$ in $P$. \qed

We can extend intuition from proposition~\ref{prop:nrtr} to graphs that contain suffix and prefix edges, but it requires the introduction of special types of walks in the graph.

\begin{definition} A walk is called \emph{open} if it consists of a single vertex or the first/last edge in the walk is not prefix/not suffix.
\end{definition}
\begin{definition} A walk is called \emph{closed} if the first/last vertices have a non-prefix incoming edge/a non-suffix outgoing edge.
\end{definition}
Note that in a graph with suffix and prefix edges, adding suffix edge at the end or prefix edge to the start of a walk does not change the label of the walk. Open and closed walks represent the shortest/longest walks with a given label.

The proposition below lists the properties of veritices in a non-redundant graph, as well as open and closed walks.

\begin{proposition}\label{prop:speprop}
    Let $\assemblygraph$ be a non-redundant assembly graph.
    \begin{enumerate}
        \item If a vertex in $\assemblygraph$ has an outgoing suffix/incoming prefix edge, it can have no other outgoing/incoming edges.
        \item If a vertex $v$ in $\assemblygraph$ is a proper prefix/suffix of vertex $u$, then there exists unique path from $v$ to $u$/from $u$ to $v$ that consists only of prefix/suffix edges.
        \item All vertices in a non-redundant graph have different labels.
        \item Vertex label can not be a proper infix of another vertex label.
        \item For any walk $P$ in $\assemblygraph$ there exists unique closed walk $P_c$ and unqiue open walk $P_o$ such that $Label(P)=Label(P_o)=Label(P_c)$. $P_o$ can be obtained from $P_c$ by removing all prefix/suffix edges from the start/end of $P_c$.
        % \item If a vertex $v$ is a prefix/suffix of $Label(P)$ for a closed walk $P$, then $P$ starts/finishes with a walk that ends/starts with $v$ and contains only prefix/suffix edges.
    \end{enumerate}
\end{proposition}
\textbf{Proof: (1)} Assume that a vertex $v$ has both a suffix and a non-suffix outgoing edge, denoted $s=vw$ and $n=vu$, respectively. The label of the path consisting solely of edge $n$ contains $Label(w)$ as a proper infix, yet does not contain $w$ as an internal subpath. This contradicts non-redundancy, since in a non-redundant graph every occurrence of $Label(w)$ must correspond to $w$ as a subpath. The remaining cases are proved analogously.

\textbf{(2)} Consider a vertex $v$ that is a suffix of a vertex $u$. By item~(1), every vertex has at most one outgoing suffix edge. Hence, there exists a unique path $P$ starting at $u$ that consists solely of suffix edges and terminates either upon reaching $v$ or when no outgoing suffix edge is available. In the former case, $P$ is the desired path.

Suppose instead that $P$ terminates before reaching $v$. Let $e$ be any edge outgoing from the terminal vertex of $P$, and consider the walk $P' = Pe$. The edge $e$ cannot be a suffix edge; therefore, $Label(v)$ is a proper infix of $Label(P')$. By non-redundancy of the graph, this implies that $v$ must occur as a subwalk of $P$, a contradiction. The remaining cases are proved analogously.

\textbf{(3)} In search of contradiction we assume that there exist vertices $u$ and $v$, that have the same label. From item 2 we can conclude that there exists a path from $u$ to $v$, such that every edge in it is a suffix edge. Since eveery consecutive vertex in this path must be shorter than the previous and the path contains at least one edge, length of $v$ must be smaller than length of $u$, a contradiction.    

\textbf{(4)} If we apply the definition of non-redundant graphs to a pair of paths, each consisting of a single vertex, we note that one of them can not be an inner subpath of another. Thus vertex label can not be a proper infix of a label of another vertex.

\textbf{(5)} Consider initial walk $P$. If we iteratively add suffix edges to the end of $P$, its label won't change. This process can not continue indefinitely since at each step the size of the last vertex reduces. The last vertex of the resulting walk does not have outgoing suffix edges, therefore it must have outgoing non-suffix edge. Similarly adding prefix edges to the front of $P$ will result in a walk that starts from a vertex that has an incoming non-prefix edge. This walk is by definition closed and by construction has the same label as the initial walk $P$. Similarly removing suffix edges from the end of $P$ and prefix edges from the start of $P$ will result in an open walk with the same label.\qed

We are now ready to prove the properties of non-redundant graphs that are satisfied even in presence of suffix and prefix edges.

\begin{proposition}\label{prop:nrsp}
    Let $\assemblygraph$ be a non-redundant assembly graph. Then:
    \begin{enumerate}
        \item For any two closed walks $P$ and $Q$ in $\assemblygraph$, such that $Label(P)$ is a substring/prefix/suffix of $Label(Q)$, $P$ is a subwalk/prefix walk/suffix walk/equal of $Q$.
        \item For a closed walk $P$, any walk $Q$, such that $Label(Q)$ is a substring of $Label(P)$ must be a subwalk of $P$.
        \item For any two open walks $P$ and $Q$ in $\assemblygraph$, where $P$ is non-trivial (contains at least one edge), $Label(P)$ is a prefix/suffix/substring of $Label(Q)$ iff $P$ is a prefix walk/suffix walk/subwalk of $Q$.
        \item For any two open walks $P$ and $Q$ in $\assemblygraph$, $Label(P)=Label(Q)$ iff $P = Q$.
        \item For any vertex or edge $v$ and closed walk $P$ in $\assemblygraph$, the number of occurrences of $v$ in $P$ is the same as the number of occurrences of $Label(v)$ in $Label(P)$.
        \item Let $C$ be a simple circuit and $d$ be the total length of its edges minus total length of all vertices. Then $d$ is the period of $Label(C)$.
    \end{enumerate}
\end{proposition}
\textbf{Proof: (1)} We will prove this statement only for substrings since the other cases can be proven analogously.
Consider the extension $Q'$ of the walk $Q$ obtained by adding one edge at the beginning and one at the end. Since $Q$ is closed, this extends the label of $Q$ at both ends. Therefore, $Label(Q)$, as well as $Label(P)$, are proper infixes of $Label(Q')$. Since the graph is non-redundant, it follows that $P$ is an inner subwalk of $Q'$. Moreover, since $Q'$ differs from $Q$ by just one edge on each side, $P$ is in fact a subwalk of $Q$.

\textbf{(2)} Consider a closed walk $Q_c$ that has the same label as $Q$. $Q$ is subwalk of $Q_c$ by item (5) of proposition~\ref{prop:speprop}. $Q_c$ is a subwalk of $P$ by item 1 in this proposition. Therefore $Q$ is a subwalk of $P$.

\textbf{(3)} Again, we prove this statement only for substrings. Consider closed walks $P'$ and $Q'$ with the same labels as $P$ and $Q$, respectively. By item (1), $P'$ is a subwalk of $Q'$, and hence $P$ is also a subwalk of $Q'$. Since $P$ is non-trivial, it must start in a non-prefix edge and end in a non-suffix edge. Therefore, removing all prefix edges from the start and all suffix edges from the end of $Q'$—which transforms $Q'$ back into $Q$—cannot remove any edges of $P$. Consequently, $P$ is a subwalk of $Q$.

\textbf{(4)} The statement for closed walks follows directly from item (1). For open walks, item (3) addresses the case where at least one of the walks is nontrivial. If both walks are trivial (i.e., consist of a single vertex), the claim follows from the fact that, in a non-redundant graph, all vertices have distinct labels.

\textbf{(5)} This proof follows the strategy of item~(2) in Proposition~\ref{prop:nrtr}, 
but additional care is required at the endpoints of the walk.

First, suppose that the vertex $v$ does not occur in the walk $P$. 
By item~(2), $v$ cannot appear as a substring of $Label(P)$. 
Hence the statement holds in this case and we can now assume that $v$ occurs at least once in $P$.

We decompose $P$ into subwalks $P = P_1 P_2 \cdots P_n$, by splitting it at each occurrence of $v$ in $P$. If $P$ starts and/or ends with $v$, we, respectively, consider $P_1$ and/or $P_n$ to be a walk consisting of a single vertex $v$.
By construction, consecutive subwalks intersect exactly at $v$. 
Consequently, the labels of the consequtive subwalks overlap by $v$. 
This establishes a bijection between occurrences of $v$ in $P$ and 
overlaps between consecutive subwalk labels. 
It remains to show that these are the only occurrences of $v$ in $Label(P)$.

Each $Label(P_i)$ has length at least $|v|$, and the overlap between 
consecutive labels is exactly $v$. Therefore, any occurrence of $v$ 
in $Label(P)$ must lie entirely within a single $Label(P_i)$ or 
span the overlap between two consecutive labels. The latter case
corresponds precisely to an occurrence already accounted for. 
Thus it suffices to rule out additional occurrences inside a single 
subwalk label.

By construction, $v$ cannot appear as a proper infix of $Label(P_i)$, 
since this would imply that $v$ occurs as an internal vertex of $P_i$, 
contradicting the definition of the decomposition. The only remaining 
possibilities are that $v$ occurs as a prefix or suffix of $Label(P_i)$.

For interior subwalks (i.e., $1<i<n$), the prefix and suffix 
occurrences coincide with the overlaps already considered. Moreover, subwalks that start and end in $v$ must contain at least one non-suffix edge (otherwise size of the last vertex would be smaller than the size of the first vertex in the walk). Therefore the size of the labels of interior subwalks are strictly larger than $|v|$, making suffix occurrence of $v$ in the subwalk different from prefix occurrence of $v$ in the same subwalk.
Hence we only need to analyze the boundary cases: prefix of $P_1$ and suffix of $P_n$.

Consider $P_1$, which is a prefix walk of $P$ that ends at $v$. 
Suppose that $Label(P_1)$ has $v$ as a prefix. We claim that this 
forces $Label(P_1)=v$. Indeed, assume instead that 
$Label(P_1)\neq v$. In this case $|Label(P_1)| > |v|$, making prefix occurrence of $v$ in $Label(P_1)$ different from suffix occurrence of $v$ in $Label(P_1)$. Since $P$ is closed, there exists an incoming 
edge $e$ to the starting vertex of $P_1$ that is not part of $P_1$. 
Extending $P_1$ by $e$ at the beginning produces a walk $Q$ whose label 
contains $v$ as a proper infix. This implies that $v$ 
occurs as an inner vertex of $Q$, which is impossible since $v$ can only occur once in $P_1$ by construction. Hence $Label(P_1)=v$, and the prefix occurrence of $v$ in $Label(v)$ is already 
accounted for by the overlap between $Label(P_1)$ and $Label(P_2)$.

The case of $P_n$ is symmetric: if $v$ is a suffix of 
$Label(P_n)$, then the whole $Label(P_n)$ must coincide with $v$, resulting in all occurrences of $v$ in $Label(P)$ being accounted for as overlaps of consequitive subwalks.

\textbf{(6)} By construction of chain label, $Label(C)[i]=Label(C)[i+d]$ for any integer $i$. We need to show that there can be no smaller $d'$ that satisfies the same condition. Let $d'$ be the period of $Label(C)$. By item (5) all occurrences of vertex and edge labels in $Label(C)$ match occurrences of corresponding verteces and edges in $C$. These occurrences of vertex and edge labels in $Label(C)$ must be periodic with period $d'$. If we collect edge and vertex label occurrences starting within one period of the $Label(C)$, we obtain a cycle $C'$, that represents the same periodic sequence of vertices and edges as $C$. Total length of edges minus total length of vertices in $C'$ is $d'$. $C$ is a simple cycle, thus $C'$ can only be longer than $C$. Therefore $d'\ge d$.\qed

This proposition allows us to prove the main property of non-redundant graphs, formulated as theorem~\ref{thm:nonredcount} in section~\ref{sec:properties}.

\setcounter{theorem}{0}
\begin{theorem}
In a non-redundant graph, for any circuit $C$ and vertex $v$, the number of occurrences of $v$ in $C$ equals the number of occurrences of $Label(v)$ in $Label(C)$.    
\end{theorem}
\textbf{Proof}.
We follow the logic of the proof of item 5 of proposition~\ref{prop:nrsp}, but border cases are not of concern in this theorem.

First, consider that $C$ does not contain vertex $v$. In such case we can consider walk $P$, that consists of $|v|+1$ complete passes through $C$, starting from an arbitrary vertex of $C$. $Label(P)$ contains all substrings of $Label(C)$ of length $|v|$ as proper infixes. $P$ does not contain vertex $v$, thus $Label(P)$ does not contain $v$ as a proper infix. Therefore, $Label(C)$ also does not contain $v$.

Let $P_1,P_2,\ldots,P_n$ be the result of splitting cycle $C$ at every occurrence of $v$ in $C$. Label of each subwalk has length larger than $|v|$ and overlaps between consecutive subwalks are exactly $v$. $v$ cannot occur as a proper infix of any of the subwalks since $v$ can not be an inner vertex for any of the subwalks. Thus all occurrences of $v$ in $Label(C)$ correspond to overlaps between consecutive subwalks, which in turn correspond to occurrences of $v$ in $C$. \qed

% \begin{lemma}\label{NONR}
% Let $P$ be a simple circuit in a non-redundant graph with vertex sequence $v_1\ldots v_n,v_{n+1}=v_1$ and edge sequence $e_1\ldots e_n$. Then there exist mappings $O_V$ and $O_E$: from 1...n to the set of positions in $Label(P)$ that map every vertex/edge to its occurrence, such that $O_V(i) = O_E(i)$ and $O_V(i+1) = O_E(i) +|e_i| - |v_{i+1}|$ and images of $O_V$ and $O_E$ cover all occurrences of vertex/edge labels in Label(P).
% \end{lemma}

% \textbf{Proof}. Since the label of $P$ is constructed as a chain label for the sequence of edge labels, mappings $O_V$ and $O_E$ exist by definition of chain label. To show that this mapping covers all occurrences of vertex/edge labels in $Label(P)$, we note that for any vertex v in the graph values of $O_V$ contain as many different occurrences of $Label(v)$ in $Label(P)$ as there are occurrences of vertex $v$ in circuit $P$. The number of occurrences of $Label(v)$ in $Label(P)$ is the same as the number occurrences of vertex $v$ in circuit $P$, thus $O_V$ defines an injection between two sets of the same size: occurrences of $Label(v)$ in $Label(P)$ and occurrences of vertex $v$ in walk $P$. Thus $O_V$ covers all occurrences of $v$ in $Label(P)$. The same is correct for any edge in the graph. Combining these conclusions together completes the proof.

\begin{proposition}
Let $P$ be a walk and $C$ be a circuit in a non-redundant assembly graph. Then, whenever $Label(P)$ is substring of $Label(C)$, $P$ is a subwalk of $C$.
\end{proposition}

\textbf{Proof}.
Let $v$ be one of the shortest vertices in $C$. Incoming and outgoing edges for $v$ in $C$ can not be prefix of suffix edges respectively. Thus the result $Q$ of cutting $C$ in one of the occurrences of $v$ is a closed walk.
If $Label(P)$ is a substring of $Label(C)$, then $Label(P)$ must also be a substring of label of walk $Q'$, constructed as walk $Q$ repeated a certain amount of times. Therefore $P$ is a subwalk of walk $Q'$, which in turn is a subwalk of $C$.\qed

Below we prove he remaining statements from section~\ref{sec:properties}.
\begin{proposition}
De Bruijn graphs and Condensed de Bruijn graphs are non-redundant when constructed from a read set that covers all $k+1$-mers in a circular genome.
\end{proposition}

\textbf{Proof}.
If the label of walk $P$ is a proper infix of the label of walk $Q$, then the sequence of $k$-mers in $Label(P)$ is a subsequence of the sequence of $k$-mers in $Label(Q)$ that does not contain the first and the last $k$-mer. Since $k$-mers from reads bijectively map to vertices of de Bruijn graph, walk $P$ is an inner subwalk of walk $Q$.\qed

\begin{theorem}
    Every weakly contig-preserving assembly graph is splitting. Every splitting, non-redundant graph is weakly contig-preserving. 
\end{theorem}
\textbf{Proof}.
First, we assume that an assembly graph $\assemblygraph$ is weakly contig-preserving. Let vertex $v$ from $\assemblygraph$ have two distinct outgoing edges $e_1$ and $e_2$. Since every edge is contained within at least one cycle, we consider the shortest cycles $C_1$ and $C_2$, containing $e_1$ and $e_2$ respectively. Note that both $C_1$ and $C_2$ contain $v$ only once since otherwise there exists a shorter cycle that contains $e_1$ or $e_2$. Finally, we consider the set of genome candidates that consists of two genome candidates: $\genomepaths=\{\{C_1\},\{C_2\}\}$. Let $P$ be the longest walk containing vertex $v$ present in all genome candidates from $\genomepaths$. $P$ must end in $v$ since $v$ is followed by different edges in different genome candidates. Since $\assemblygraph$ is weakly contig-preserving, $Label(P)$ must be a maximal string that is a substring of all genome candidates. Assume that $e_1$ and $e_2$ have a common prefix $s$, that is longer than $v$. In this case any label of a genome candidate contains not just $Label(P)$, but also a combined sequence of $Label(P)$ and $s$. This combined sequence contains $Label(P)$ as a proper substring, which contradicts the maximality of $Label(P)$. Therefore $e_1$ and $e_2$ must diverge immediately after $v$ and $\assemblygraph$ is splitting.

Let us assume now that the graph $\assemblygraph$ is splitting and non-redundant. Consider a walk $P$ that is one of the maximal walks, present in every genome path candidate from $\genomepaths$. $Label(P)$ is naturally present as a substring of every genome path candidate label, but let's assume that it is not maximal such string. In particular, there exists a string $s$, that contains $Label(P)$ as a proper substring that is also present in every label of genome path candidate. Without loss of generality we will assume that $Label(P)$ is a proper prefix of $s$. Consider a chromosome $C$ of a genome path candidate from $\genomepaths$, such that $s$ is a substring of $Label(C)$. Every occurrence of $s$ in $Label(C)$ is also an occurrence of its prefix $Label(P)$ in $Label(C)$, which in turn corresponds to an occurrence of $P$ in $C$ as a subwalk. Consider the next edge $e$ after $P$ in the occurrence of $P$ in $C$ that matches the occurrence of $s$ in $Label(C)$. Note that edge $e$ will be the same for any genome candidate we consider, since $e$ is an edge outgoing from the last vertex $v$ of $P$ and the first nucleotide in $Label(e)$ after $v$ is defined by $s$. Therefore walk $P+e$ is present in every genome candidate, which contradicts the maximality of $P$. Therefore for non-redundant splitting graphs, labels of all walks from $OGA(\genomepaths)$ are maximal contigs, that are present in every genome path label.\qed.

\begin{theorem}%\label{thm:limitation}
For any non-redundant splitting assembly graph $\assemblygraph$ and a set of genome path candidates $\genomepaths$, $$\Labels(\OGA(\genomepaths))=\Superstrings(\OA(\Labels(\genomepaths)), V(\assemblygraph)).$$.
\end{theorem}
\textbf{Proof}.
By theorem~\ref{thm:cprescond} we know that $\Labels(\OGA(\genomepaths))\subset \OA(\Labels(\genomepaths))$. Every walk label from $\Labels(\OGA(\genomepaths))$ contains at least one vertex label as a substring, thus we can conclude $$\Labels(\OGA(\genomepaths))\subset\Superstrings(\OA(\Labels(\genomepaths)), V(\assemblygraph)).$$
Consider $s\in \Superstrings(\OA(\Labels(\genomepaths)), V(\assemblygraph))$, that is a maximal contig present in every label of genome path candidate, that also contains at least one vertex as a substring. Consider the longest walk $P$, such that $Label(P)$ is a substring of $s$. Such walk exists since $s$ at least contains a label of one vertex as a substring. Since $Label(P)$ is present in every label of genome path candidate as a substring, walk $P$ must be present in every genome path candidate as a subwalk. We thus have the same situation as in the proof of theorem~\ref{thm:cprescond}: if $Label(P)\neq s$, we can extend $P$ to find a longer walk $P'=P+e$ that is present in every genome path candidate. If $Label(P')$ is still a substring of $s$, we have a contradiction with maximality of $P$. If $Label(P')$ is not a substring of $s$, we can combine $s$ and $Label(P')$ to construct a longer string $s'$ that is present in every genome path candidate label as a substring, which contradicts the maximality of $s$. Therefore $Label(P)=s$ and $P$ is a maximal walk that is present in every genome path candidate as a subwalk.\qed

\appsection{Chromosome candidate preservation and conductor strings.}\label{app:conductor}
We start by proving theorem~\ref{thm:conventional} from section~\ref{sec:conductors}.
\begin{theorem}
For read set $\reads$
$\CPC(\OG(\reads))=\CPC(\SG(\reads))\subseteq\CPC(\ROG(\reads))=\CC(\reads)\subseteq\CPC(\DBG(\reads,k))$. Moreover, each of these inclusions is strict for some read sets.
\end{theorem}

\textbf{Proof}. Raw overlap graph features a bijection between reads and vertices. Its edges represent overlaps. Thus there is a bijection between cyclic read chains and circuits in raw overlap graphs, constructed from these reads. This bijection preserves labels, thus the set of read chain labels is the same as the set of cycle labels in $ROG(\reads)$. Overlap graph is the result of removing some vertices and edges from $ROG(\reads)$. Thus the set of circuits in $OG(\reads)$ is a subset of the set of circuits in $ROG(\reads)$. Since circuit-to-chain mapping still preserves labels, we have $\CPC(\SG(\reads))\subseteq\CPC(\ROG(\reads))$. String graph is obtained from Overlap graph by removing transitive edges. Transitive edges do not change the set of circuit labels since for any circuit that contains transitive edges, one can construct a circuit that does not contain them and has the same label. Therefore $\CPC(\OG(\reads))=\CPC(\SG(\reads))$. Finally any read $k$-chain label can be represented as a sequence of $k$ and $k+1$-mers, each coming from one of the reads. Thus de Bruijin graph constructed from these reads will contain a walk with the same label as any label of a read $k$-chain. therefore $\CC(\reads) \subseteq\CPC(\DBG(\reads,k))$.

Examples of cases where inclusions are strict can be found in fogure~\ref{fig:assembly-graphs}.
\qed

\paragraph{Properties of conductor strings.}
The definition of conductor strings uses self-overlapping strings that start and finish with conductor string $s$. Many proofs would be easier if it was possible to record such strings as $sxs$ for some string $x$. Unfortunately it is not always possible since a self-overlapping strings can be shorter than double length of $s$. For example, string ATATA both starts and finishes with ATA, but can not be represented as ATAxATA. That is why when we describe self-overlapping string $t$ we introduce two strings $x$ and $x'$, such that $t=sx=x's$. The following lemma shows an important connection between $x$ and $x'$.

\begin{lemma}\label{lm:cyclicswitch}
    For any finite strings $x,x'$, if there exists string $s$, such that $sx=x's$ then $\langle x\rangle=\langle x'\rangle$
\end{lemma}
\textbf{Proof}. We will prove that $a$ represents a cyclic shift of $b$ by $|s|$ positions. Indeed, comparing letters at position $n$ in $sa$ and $bs$ we conclude that for $|s|\le n\le |a|$, letter at position $n-|s|$ in $a$ is the same as the letter at position $n$ in $b$. For positions $n<|s|$ the letter at position $n$ in $s$ matches letter at position $n$ in $b$, which in turn (by comparison of letters at positions $|a| + n$ in $sa$ and $bs$) matches the letter at position $n+|a|-|s|$ in $a$. \qed

The next lemma helps us characterize substrings of cyclic strings.
\begin{lemma}\label{lm:cyclicsubstrings}
    If a finite string $s$ is a substring of a cyclic string $\langle t\rangle$, then there exists string $w$, such that $\langle t\rangle=\langle sw\rangle$
\end{lemma}
\textbf{Proof}.
Since $s$ is a finite substring of a string obtained as infinite repetition of string $t$, $s$ must be a substring of a finite repretition $t^n$ of string $t$. In other words $xsy=t^n$. $\langle t\rangle=\langle t^n\rangle=\langle xsy\rangle=\langle syx\rangle$, thus the statement of the lemma is satisfied for $w=yx$.
\qed

In the following we prove important properties of conductor strings and use them to prove the remaining theorems from section~\ref{sec:conductors}. 

\begin{lemma}\label{lm:switch}
Let $s$ be a conductor string with respect to a chromosome-set $\chromosomes$. Then for any pair $a_1sb_1,a_2sb_2$ of superstrings of $s$, that are substrings of $\chromosomes$, the results of the suffix switch $a_1sb_2$ and $a_2sb_1$ are also substrings of $\chromosomes$.
\end{lemma}
\textbf{Proof}. Let $w_1=\langle a_1sb_1t_1\rangle=\langle sb_1t_1a_1\rangle$ and $w_2=\langle a_2sb_2t_2\rangle=\langle sb_2t_2a_2\rangle$ be members of $\chromosomes$ that contain $a_1sb_1$ and $a_2sb_2$ as substrings respectively (by Lemma~\ref{lm:cyclicsubstrings}. Since $s$ is conductor, we can apply definition of conductor strings to $s$ and strings $sbt_1as$ and $sdt_2cs$ that have $s$ as both prefix and suffix and conclude that $\langle b_1t_1a_1sb_2t_2a_2s\rangle=\langle b_2t_2a_2sb_1t_1a_1s\rangle$ also belong to $\chromosomes$. Therefore $a_1sb_2$ and $a_2sb_1$ are also substrings of $\chromosomes$.\qed

\begin{lemma}\label{lm:circularization}
Let $s$ be a conductor string with respect to chromosome-set $\chromosomes$ and $sa=a's$ be a substring of $\chromosomes$ that starts and finishes with $s$. Then $\langle a\rangle\in\chromosomes$.
\end{lemma}
\textbf{Proof}. Let $\langle w\rangle$ be a cyclic string from $\chromosomes$ that contains $sa$ as a substring. Then by lemma~\ref{lm:cyclicsubstrings} $\langle w\rangle=\langle sat\rangle=\langle ats\rangle$. Consider strings $sa$ and $sts$, that have $s$ as both prefix and suffix. Since $s$ is a conductor string and $\langle ats\rangle\in\chromosomes$, we conclude that $\langle a\rangle, \langle st\rangle\in\chromosomes$.\qed

\begin{theorem}
For any non-redundant assembly graph $\assemblygraph$, the following strings are conductors for $\chromosomes = \CPC(\assemblygraph)$:
(i) any superstring of a conductor;
(ii) any vertex, edge, or walk label in $\assemblygraph$;
(iii) any string that is not a substring of $\chromosomes$;
(iv) any string that is at least as long as the longest edge label in $\assemblygraph$.
\end{theorem}
\textbf{Proof}.
\textbf{(i)} Assume that $s$ is a conductor string and consider its superstring $t=xsy$.
Consider an arbitrary pair of strings $ta=xsya$, $tb=xsyb$ that start and end with $t$. Since $t$ ends with $y$ we can also rewrite these as $ta=xsya=xsa'y$, $tb=xsyb=xsb'y$ for $a',b'$ such that $ya=a'y$ and $yb=b'y$. Since $yab=a'yb=a'b'y$, we have $\langle ab\rangle=\langle a'b'\rangle$ by lemma~\ref{lm:cyclicswitch}. Therefore $\langle ab\rangle\in\chromosomes$ iff $\langle a'b'\rangle\in\chromosomes$. Since both $ta=xsa'y$ and $tb=xsb'y$ finish with $t=xst$, we have $sa'$ and $sb'$ both start and finish with $s$. We apply the definition of conductor string to $s$ to obtain $\langle a'b'\rangle\in\chromosomes$ iff $\langle a'\rangle,\langle b'\rangle\in\chromosomes$. From $ya=a'y$ and $yb=b'y$, we have $\langle a\rangle=\langle a'\rangle$ and $\langle b\rangle=\langle b'\rangle$ by lemma~\ref{lm:cyclicswitch}. Therefore $\langle a'\rangle,\langle b'\rangle\in\chromosomes$ iff $\langle a\rangle,\langle b\rangle\in\chromosomes$. Combining all observations, we conclude that $\langle ab\rangle\in\chromosomes$ iff $\langle a\rangle,\langle b\rangle\in\chromosomes$. Thus $t$ (arbitrary superstring of a conductor string $s$) is a conductor string.

\textbf{(ii)} We will prove the statement only for vertex labels, since walk and edge labels are superstrings of vertex labels. Consider a vertex $v$ and strings $va$ and $vb$, that both start and end with $v$. Assume that $\langle a\rangle, \langle b\rangle\in\chromosomes$. Thus there exist cycles $C_1, C_2$, such that $Label(C_1)=\langle a\rangle$ and $Label(C_2)=\langle b\rangle$. Since all occurrences of a vertex $v$ in a cycle correspond to occurrences of $v$ in its label we can find occurrences of $v$ that correspond to the start and end of substring $va$ in the $Label(C_1)$ and $vb$ in $Label(C_2)$ in $Label(C)$. The walks $P_1$ and $P_2$, defined as subwalks of $C_1$ and $C_2$ between these pairs of occurrences of $v$ have labels $sa$ and $sb$ correspondingly. Since $P_1$ and $P_2$ both start and end in $v$ we can combine them into a single cycle $C$ that goes through $P_1$ first and through $P_2$ after that. $Label(C)=\langle ab\rangle$. Therefore $\langle ab\rangle\in\CPC(\assemblygraph)=\chromosomes$. The opposite implication from the conductor string definition is proven analogously.

\textbf{(iii)} If a string $v$ is not a substring of $\chromosomes$, then no superstring of $v$ can belong to $\chromosomes$. Therefore the bidirectional in the definition of conductor strings is always correct as both left side $\langle ab\rangle\in\chromosomes$ and the right side $\langle a\rangle,\langle b\rangle\in\chromosomes$ are always false.

Note that similarly one can show that cyclic strings are conductors since cyclic strings have no superstrings.

\textbf{(iv)} Consider a string $s$ that is at least as long as the longest edge label in $\assemblygraph$. If $s$ is not a substring of $\chromosomes$, it is a conductor based on item (iii). If $s$ is a substring of $\chromosomes$, it must be a substring of a cycle lablel, which is a combined sequence of a sequence of edges. It can not be strictly contained within a single edge label, thus it must contain at least one vertex label as a substring. Since every vertex label is a conductor and all superstrings of conductors are also conductors $s$ is also a conductor.\qed

\begin{theorem}%\label{thm:criterion}
If there exists $K$, such that any string of length $K$ is a conductor with respect to chromosome-set $\chromosomes$, then $\CPC(DBG(\chromosomes, K))=\chromosomes$.
\end{theorem}
\textbf{Proof}. By theorem~\ref{thm:conventional} $\CC(\chromosomes) \subseteq\CPC(DBG(\chromosomes,K))$. Since cyclic strings can form only trivial chains consisting of a single cyclic sequence, $\CC(\chromosomes)=\chromosomes$ . Therefore $\chromosomes \subseteq \CPC(DBG(\chromosomes,K))$ and it is sufficient to show that $\CPC(DBG(\chromosomes,K)) \subseteq \chromosomes$. We will use induction to show that any walk label in $\CPC(DBG(\chromosomes,K))$ is a substring of $\chromosomes$. Any path of length $1$ is a substring of $\chromosomes$ since it corresponds to a $k+1$-mer in $\chromosomes$. Let $P$ and $Q$ be two walks that end and start in vertex $v$ respectively. By induction assumption $Label(P)$ and $Label(Q)$ are substrings of $\chromosomes$. Since $v$ is a conductor, by lemma~\ref{lm:switch},  $Label(PQ)$ is also a substring of $\chromosomes$. This shows that the label of any walk is a substring of $\chromosomes$. Consider a circuit $C$ and a vertex $v$ in it. If we cut circuit $C$ in vertex $v$, the resulting walk $P$ starts and ends in a vertex $v$ and its label is a substring of $\chromosomes$, that starts and ends with a conductor string $v$. By lemma~\ref{lm:circularization} we can conclude that $Label(C)$ belongs to $\chromosomes$ for any circuit $C$, which concludes the proof.\qed

The following lemma provide the necessary tools to apply the criterion to the set of chain labels.

\begin{lemma}\label{lm:chain_prefix}
Let $P$ be a read chain and $s$ a prefix of $Label(P)$. Let $r$ be a non-contained read, such that $r$ is a suffix of $s$. Then there exists a read chain $Q$, that ends with read $r$, coinsides with $P$ before the last read and $Label(Q)=s$.
\end{lemma}
\textbf{Proof}.
Without loss of generality we consider only the case Consider the longest prefix chain $P'$ of $P$, such that $|Label(P')|\le|s|-|r|$. Let the next read after prefix chain $P'$ in $P$ be $w$. Combined sequence of $P'$ and $w$ has length larger than $|s|-|r|$ and is a prefix of $s$. Thus $w$ has overlap with $r$, such the label of a chain $Q$, that consists of $P'$, $w$, and $r$ is $s$. Chain $Q$ thus satisfies the required conditions.    
\qed

Note that the same argument can be applied in a reverse situation where $s$ is a suffix of $Label(P)$, that starts with the sequence of read $r$.

\begin{theorem}
If $\reads$ is a read set with maximal read length $L$, then every string of length $2L$ and every non-contained read is a conductor with respect to $\chromosomes=\CC(\reads)$.
\end{theorem}

\textbf{Proof}. Any string $s$ that has length at least $2L$ must be either not present as substring in any chain (in which case it is a conductor by theorem~\ref{thm:conductors} item~iii) or contains a non-contained read as a substring. Since superstrings of conductors are conductors themselves, it is sufficient to prove the theorem for non-contained reads.

Let $r$ be a non-contained read. Consider strings $ra_1=b_1r, ra_2=b_2r$, that start and finish with $r$. Assume $\langle a_1\rangle, \langle b_1\rangle\in\chromosomes$. Then there exist read chains $C_1$ and $C_2$, such that $Label(C_1)=\langle a_1\rangle=\langle b_1\rangle$ and $Label(C_2)=\langle a_2\rangle=\langle b_2\rangle$. Consider read chains $P_1$ and $P_2$, that represent sufficient amount of repeated units of cyclic chains $C_1$, $C_2$ to contain $ra_1$ and $ra_2$ as label substrings respectively. Applying lemma~\ref{lm:chain_prefix} twice (once for prefix and once for suffix) we can conclude that there exist chains $Q_1, Q_2$ that start and end with $r$ and have labels $ra_1$ and $ra_2$ respectively. We can combine these two chains into a single larger chain with label $\langle a_1a_2\rangle$. Thus $\langle a_1a_2\rangle\in\chromosomes$. The reverse implication (if $\langle a_1a_2\rangle\in\chromosomes$, then $\langle a_1\rangle,\langle a_2\rangle\in\chromosomes$ can be proven analogously.\qed

\appsection{Universal assembly graph construction}\label{app:universal}

\begin{theorem}%\label{thm:ag}
\begin{enumerate}
 \setlength{\itemsep}{0pt}   % no extra space between items
  \setlength{\parskip}{0pt}   % no paragraph spacing in items
  \setlength{\topsep}{0pt}    % no space above/below list
  \setlength{\parsep}{0pt}    % no extra spacing between paragraphs in item
  \setlength{\partopsep}{0pt} % no extra spacing when starting a new paragraph
  \item For any non-redundant assembly graph $\assemblygraph$, we have
    $AG(\CPC(\assemblygraph), V(\assemblygraph)) = \assemblygraph$.
    \item For any infix-free string-set $v$ and chromosome-set $\chromosomes$ $AG(\chromosomes,V)$ is a non-redundant graph.
    \item Let $\chromosomes$ be a cyclic string-set, and let $V$ be a proper infix-free set of strings such that every string in $\chromosomes$ contains at least one element of $V$ as a substring. Then $\chromosomes \subseteq \CPC(AG(\chromosomes, V))$, and for any other non-redundant assembly graph $\assemblygraph$ satisfying $\chromosomes \subseteq \CPC(\assemblygraph)$ and $V(AG(\chromosomes, V)) \subseteq V(\assemblygraph)$, we have $\chromosomes \subseteq \CPC(AG(\chromosomes, V)) \subseteq \CPC(\assemblygraph)$.
\end{enumerate}
\end{theorem}

\textbf{Proof: (1)} In non-redundant assembly graphs the number of occurrences of vertex labels in a cycle label matches the number of occurrences of the vertex in the cycle. Thus a mapping that transforms vertices in the cycle into positions where they occur in the label is an order-preserving byjection. Therefore the cycle extracted from $Label(C)$ for a circuit $C$ using algorithm $AG$ matches $C$. Since every vertex and edge in assembly graph is covered by at least one cycle, gluing all cycles in the graph $\assemblygraph$ through matching edge and vertex sequences results in complete graph $\assemblygraph$.

\textbf{(2)} Consider a closed walk $P$ in graph $AG(\chromosomes,V)$. We will show that the sequence of vertices in $P$ matches the sequence of occurrences of strings from $V$ in $Label(P)$. More precisely, we will prove that for any prefix $s$ of string $Label(P)$ that ends with vertex $v\in V$, there exists a prefix walk $Q$ of $P$, that ends with $v$ and $Label(Q)=s$. Consider the shortest prefix walk $Q_1$ of $P$, such that $|Label(Q_1)|\ge|s|$ (it exists because $Label(P)>|s|$). Consider the last edge $e$ from $u$ to $w$ in the walk $Q_1$ and walk $Q_0$ that is the result of removing $e$ from $Q_1$. Assume that $v$ is not a substring of $e$. Then $|Label(Q_1)|-|e|=|Label(Q_0)|-|u|> |s|-|v|$. On the other hand $|Label(Q_0)|<|s|$, which means that vertex $v$ contains vertex $u$ (last vertex in $Q_0$) as a proper infix, which is impossible since they are both from proper infix free set $V$. Therefore $v$ is a substring of $e$.

$v$ can not be a proper infix of $e$, since otherwise the sequence of $v$ will occur between $u$ and $w$ in the chromosome $e$ was extracted from. For the same reason $v$ can be neither a prefix of $e$ larger than $u$ nor a suffix of $e$ that is larger than $w$. Thus $v$ is either a prefix of $u$ or a suffix of $w$. The first case is impossible since it would mean that $|Label(Q_0)|\ge|s|$. Therefore the only remaining case to consider is that $v$ is a suffix of $w$.

Consider all proper suffixes $w_1,w_2,\ldots,w_n$ of $w$ that are present in $V$ and are longer than $v$. Any chromosome sequence $c$ that contains $w$ also contains $v$ as well as $w_1,w_2,\ldots,w_n$. When vertex occurrences in $c$ are transformed into a sequence, $w$ is always followed by $n+1$ vertices, connected to $w$ by suffix edges from $w$ to $w_1$, from $w_1$ to $w_2$, etc., until it reaches $v$ with edge from $w_n$ to $v$. Each of these edges is the only outgoing edge for its start, thus walk $P$ has no choice but to follow this path after prefix walk $Q$. Also none of these edges can be the end of walk $P$, since $P$ is a closed walk. Therefore walk $P$ contains all of these edges. Adding these $n+1$ edges to prefix walk $Q$ of $P$ results in a prefix walk of $P$ that has label $s$ (adding suffix edges to the end of the walk does not change its label) and ends with vertex $v$.

From this we infer that every occurrence of a vertex label in a closed walk $P$ matches an occurrence of a vertex at corresponding position. To complete the proof of non-redundancy of graph $AG(\chromosomes,V)$ we consider a walk walk $Q$ such that $Label(Q)$ a proper infix of $Label(P)$ for a walk $P$. Let $P_c$ and $Q_c$ be closed walks that have the same labels as $P$ and $Q$ correspondingly. The seqeuence of occurrences of vertex labels in $Label(Q)$, that corresponds to actual vertices of $Q_c$ is also present in $P_c$ since $Label(Q_c)$ is a substring of $Label(P_c)$. Therefore (from the property proven above), $Q_c$ is a subwalk of $P_c$, whose label is a proper infix of $Label(P_c)$. $P$ differs from $P_c$ by adding a few prefix edges to the start and suffix edges to the back. Removing these edges would not change that $Q_c$ and thus $Q$ is a inner substring of $P$. Thus for any pair of walks $P, Q$, such that $Label(Q)$ is a proper infix of $Label(P)$, we conclude that $Q$ is inner subwalk of $P$, showing that $AG(\chromosomes,V)$ is non-redundant.
% Rewrite this proof. Do it in terms of positions an segments instead of substrings. Fix the problem that closed walks were discussed only in context of non-redundant graphs. Create a separate statement that proves equivalency of default definition of redundant graph and one with closed walks.

\textbf{(3)} Construction of $AG(\chromosomes, V)$ involves transformation of each chromosome candidate $\langle s\rangle$ from $\chromosomes$ into a circuit $C$, where vertices correspond to occurrences of strings from $V$ in $\langle s\rangle$ and edges are extracted to connect consecutive vertices. The label of $C$ is $\langle s\rangle$. After all vertices and edges with matching sequences are glued together, circuit $C$ maps to a circuit $C'$in $AG(\chromosomes, V)$, with a mapping that preserves edge and vertex sequences. Thus it also preserves the label of $C$ and $Label(C')\in\CPC(AG(\chromosomes, V))$ for any $C$ from $\chromosomes$, and therefore $\chromosomes \subseteq \CPC(AG(\chromosomes, V))$.

Next we consider a non-redundant assembly graph $\assemblygraph$ such that $\chromosomes \subseteq \CPC(\assemblygraph)$ and $V(AG(\chromosomes, V)) \subseteq V(\assemblygraph)$.
We will prove that for any edge $e$ from $u$ to $v$ in $AG(\chromosomes, V)$, graph $\assemblygraph$ has a walk $P$ from $u$ to $v$, such that $Label(e) = Label(P)$. By construction of $AG(\chromosomes, V)$, $e$ is present in one of the cycles derived from chromosomes. We denote the chromosome, from which the cycle containing $e$ is derived as $C$. Since $\chromosomes \subseteq \CPC(\assemblygraph)$, there exists a circuit $C'$ with the same label as $C$. Since $V(AG(\chromosomes, V)) \subseteq V(\assemblygraph)$, each vertex in $C$ matches an occurrence of a vertex with the same label in $C'$. In particular, we consider occurrences of vertices $u$ and $v$, that flank the occurrence of $e$ in $C$. The label of the walk segment (in graph $\assemblygraph$) between these occurrences matches the label of $e$ (in graph $AG(\chromosomes, V)$). Therefore every edge $e$ in $AG(\chromosomes, Vertices)$ corresponds to a walk in graph $\assemblygraph$, connecting vertices with the same labels and itself having the same label as $e$. Thus every cycle in $AG(\chromosomes, V)$ can be translated into a cycle in graph $\assemblygraph$ by replacing each edge with a walk having the same label. This transformation preserves the label of the cycle. Therefore $\CPC(AG(\chromosomes, V))\subseteq\CPC(\assemblygraph)$.

\appsection{Multiplicity-extremal substrings}\label{app:mes}
\begin{theorem}
A finite string $s$ is MES for a chromosome-set $\chromosomes$ if and only if there exist characters $c_1\neq c_2$ and $c_3\neq c_4$, such that $c_1s$, $c_2s$, $sc_3$, and $sc_4$ are substrings of $\chromosomes$.
\end{theorem}

\textbf{Proof}. Let $s$ be a MES for a chromosome-set $\chromosomes$ and $sc_3$ be one of its occurrences in $\chromosomes$. Since $s$ is MES, not every occurrence of $s$ in $\chromosomes$ is followed by $c_3$. Therefore there exists an occurrence of $s$, followed by a different character ($c_4$). Thus there exist $c_3$ and $c_4$, such that $sc_3$ and $sc_4$ are both substrings of $\chromosomes$. Proof for existence of $c_1$ and $c_2$ is analogous.

Let now $s$ be a string such that $c_1s$, $c_2s$, $sc_3$, and $sc_4$ are substrings of $\chromosomes$. Let's consider an extension $sc$ of $s$ by one nucleotide. Without loss of generality we can assume that $c\neq c_3$. Since $sc_3$ is a substring of $\chromosomes$, it is impossible that every occurrence of $s$ extends to an occurrence of $sc$. Similarly we can prove that left extensions have less occurrences than $s$ in $\chromosomes$ as well, which means that $s$ is MES.

\begin{lemma}\label{lm:mem_mes}
If a string $s$ has MEM between two strings $r_1, r_2$ (common substring, that can not be extended to a larger match), then $s$ is also MES with respect to the string-set $\{r_1,r_2\}$.
\end{lemma}

\textbf{Proof}. Consider MEM between strings $r_1$ and $r_2$ that has sequence $s$. Since it is MEM, it can not be extended to a longer match. Therefore, any extention of $s$ will lose at least one of the two occurrences, constituting MEM. Therefore $s$ is a MES with respect to the string-set $\{r_1,r_2\}$.\qed

\begin{lemma}\label{lm:cyclic_mes}
Any cyclic string from $\chromosomes$ is MES with respect to $\chromosomes$.
\end{lemma}

\textbf{Proof}. Since cyclic strings have no extensions, any extension of a cyclic string has less occurrences in $\chromosomes$.\qed.

\begin{lemma}\label{lm:overlap}
    MEM between two MESs with respect to a string-set $\chromosomes$ is also MES.
\end{lemma}
\textbf{Proof}.
Let $a_1sb_1, a_2sb_2$ the two MESs in question and $s$ represent the sequence of the MEM. Since $s$ is a match between these two strings that can not be extended further, then 
\qed.

\begin{theorem}%\label{thm:mesext}
For any substring $s$ of a chromosome-set $\chromosomes$, there exists a unique MES, denoted $MES(s)$ (\emph{MES extension} of $s$), such that $s$ occurs exactly once in $MES(s)$ as a substring and every occurrence of $s$ in $\chromosomes$ extends to an occurrence of $MES(s)$.
\end{theorem}

\textbf{Proof}.
Consider an occurrence $(p_1,s)$ of a string $s$ in a chromosome $c_1$. An occurrence $(q_1, m_1)$ of a minimal MES $m_1$ in the same chromosome $c_1$, that covers occurrence $(p,s)$. In other words there is a position $r_1$ in $m_1$, such that $(q_1,m_1)\circ(r_1,s)=(p_1,s)$. Similarly we consider an occurrence $(p_2,s)$ of $s$ in a chromosome $c_2$ and minimal MES $m_2$ that covers this occurrence of $s$ in $c_2$: $(q_2,m_2)\circ(r_2,s)=(p_2,s)$. $s$ is a match between $m_1$ and $m_2$. It can be extended to a MEM with string $s'$: for position $t$ of $s$ in $s'$ and positions $r_1', r_2'$ of $s'$ in $m_1$ and $m_2$, we have $(t,s)\circ (r_1',s')=(r_1,s)$ and $(t,s)\circ (r_2',s')=(r_2,s)$.

If either $r_1'$ or $r_2'$ is a start position in $m_1$ or $m_2$, then there is a MES that starts with $s'$, which means (by theorem~\ref{thm:mesalt}) that there exist $c_1,c_2$, such that $c_1s$ and $c_2s$ are substrings of $\chromosomes$. Similarly if both $r_1'$ are $r_2'$ start positions in $m_1$ or $m_2$, the letters preceding positions $r_1'$ and $r_2'$ in $m_1$ and $m_2$ must be different since since $s'$ is constructed as the longest extension of a match $s$ between $m_1$ and $m_2$. Therefore if we set $c_1$ and $c_2$ to be the letters in positions $r_1-1$ and $r_2-1$ of $m_1$ and $m_2$ correspondingly, we again have $c_1s$ and $c_2s$ are substrings of $\chromosomes$. Similarly, by considering the end of match $s'$ between $m_1$ and $m_2$, we can find $c_3$ and $c_4$, such that $sc_3$ and $sc_4$ are substrings of $\chromosomes$. By theorem~\ref{thm:mesalt} we conclude that $s'$ is a MES, that covers the considered occurrences of $s$ in $c_1$ and $c_2$. Since $m_1$ and $m_2$ are minimal such MES, $m_1=s'=m_2$ and $r_1=r_2=t$. Thus any occurrence of $s$ in chromosomes can be extended to an occurrence of a MES and the same chromosome. 

If we assume that $s$ occurs at least twice in $MES(s)$ at positions $r_1$ and $r_2$ we can use the same argument to conclude that $r_1=r_2$ and reach a contradiction.\qed

\begin{lemma}\label{lm:mes_substring}
Whenever $s$ is a substring of $t$, the string $MES(s)$ is a substring of $MES(t)$. Furthermore, if $s$ is a conductor as well as MES, Whenever $s$ is a prefix/suffix of $t$, the string $MES(s)$ is a prefix/suffix of $MES(t)$.
\end{lemma}

\textbf{Proof}. Since $MES(t)$ contains $t$ as a substring, it also contains $s$ as a substring and thus by theorem~\ref{thm:mesext}, $MES(s)$ is a substring of $MES(t)$.

Let $m=MES(s)$ and $MES(t)=amb$ B. We will show that in case $s$ is a conductor string $mb$ is a MES. By theorem~\ref{thm:mesalt}, applied to $m$, 
there are $c_1,c_2$ such that $c_1m$ and $c_2m$ are substrings of $\chromosomes$. By lemma~\ref{lm:switch}, applied to conductor string $s$ and substrings of $\chromosomes$ $mb$, $c_1m$, $c_2m$, we conclude that $c_1mb$ and $c_2mb$ are also substrings of $\chromosomes$. By theorem~\ref{thm:mesalt}, applied to string $amb$, there are $c_3$ and $c_4$, such that $ambc_3$ and $ambc_4$ are substrings of $\chromosomes$, and in particular $mbc_3$ and $mbc_4$ are substrings of $\chromosomes$ as well. Therefore by theorem~\ref{thm:mesalt}, applied to $mb$, we conclude that $mb$ is a MES. In case $s$ is a prefix of $t$, $mb$ contains $t$ as a substring. Thus $a$ is empty and $m$ is a prefix of $MES(t)$. The case when $s$ is a suffix of $t$ can be proven the same way.\qed

% \appsection{Chain labels as a language class}\label{app:subclass}
% Chains of overlapping reads can be considered as a class of languages. Each language is charactarized by a finite set $W$ of words and the language itself consists of labels of chains of words from the set $W$. To stay within the context of cyclic strings we consider only labels of cyclic chains. Similarly, every assembly graph defines a language, that consists of labels of its circuits. In section~\ref{sec:conductors} we have shown that for any set $W$ there exists an assembly graph $\assemblygraph$, such that the chain label language defined by $W$ is exactly the same as the seet of circuit labels in $\assemblygraph$. Here we will show that the opposite is incorrect: there is an assembly graph $\assemblygraph$, such that labels of its circuts cannot be represented as chains labels for any set of words.

\appsection{Supregraphs}\label{app:supregraphs}
Note that supregraphs can not have parallel edges. Therefore we will sometimes denote edges in supregraphs by a pair of its start and end vertices.

\begin{lemma}\label{lm:condmes}
If string $s$ has a prefix and suffix that are conductor MES, then $s$ is also a conductor MES.
\end{lemma}

\textbf{Proof}. As a superstring of a conductor, string $s$ is automatically conductor. To show that it is MES, we apply theorem~\ref{thm:mesalt} to its conductor MES prefix $a$ and suffix $b$ to find letters $c_1\neq c_2$ and $c_3\neq c_4$, such that $c_1a$, $c_2a$, $bc_3$, $bc_4$ are substrings of $\chromosomes$. We then apply lemma~\ref{lm:switch} to conductors $a$ and $b$ and chromosome substrings $c_1a,s$, $c_2a$, $bc_3$, $bc_4$, paired with $s$ to conclude that $c_1S$, $c_2S$, S$c_3$, S$c_4$ are substrings of $\chromosomes$, which my theorem~\ref{thm:mesalt} implies that $s$ is a MES.\qed

\begin{theorem}
Let $\chromosomes$ be a finitely-conducting chromosome-set. Then $\CPC(AG(\chromosomes, \Core(\chromosomes)))=\chromosomes$. Each edge in $AG(\chromosomes, \Core(\chromosomes))$ is either a prefix or a suffix edge.
\end{theorem}

\textbf{Proof}. Consider consecutive occurrences of core sequences $a, b$ in a chromosome candidate and an edge sequence $e$ that connects them. We will show that $e$ is also a core sequence. By lemma~\ref{lm:condmes} $e$ is a conductor MES. Any conductor MES, representing a proper infix of $e$, would have an occurrence between $a$ and $b$ in the chromosome candidate, which is impossible since $a$ and $b$ are consecutive core sequence occurrences. Therefore $e$ is a conductor MES that does not contain other conductor MES as proper infixes. In this case $e$ should have also been marked as a core sequence occurrence between $a$ and $b$, unless sequence of $e$ equals either $a$ or $b$, which means that every edge $e$ in the graph is either a prefix or a suffix edge.

Let $\assemblygraph=AG(\chromosomes, \Core(\chromosomes))$.
By item 3 of theorem~\ref{thm:ag} we have $\chromosomes\subseteq \CPC(\assemblygraph)$, so we only need to prove the opposite inclusion. We start by proving by induction that label of every walk $P$ in $\assemblygraph$ is a substring of $\chromosomes$. The base case is for a walk, consisting of a single vertex. Every vertex is a core string, and thus is a substring of $\chromosomes$. Note since every edge is either prefix or suffix, edge labels are also substrings of $\chromosomes$. Consider a walk $P$ and edge $e$ that extends it. Both $Label(P)$ and $e$
 are substrings of $\chromosomes$. The last vertex is $P$ is a conductor that is also a prefix of $e$. Thus by lemma~\ref{lm:switch}, applied to $Label(P)$ and $e$ we infer that the extension of $P$ with $e$ also has a label that is a substring of $\chromosomes$.

 Consider an arbitrary circuit $C$ in $\assemblygraph$. Let $v$ be an arbitrary vertex in $C$ and $P$ be a walk, resulting from breaking $C$ in vertex $v$. $Label(P)$ is a substring of $\chromosomes$ that starts and ends with conductor string $v$. Thus by lemma~\ref{lm:circularization} we can conclude that $Label(C)\chromosomes$. Therefore $\CPC(\assemblygraph)\subseteq\chromosomes)$.\qed

The following proposition provides properties of condensed supregraphs and explains why it is called condensed.
\begin{proposition}\label{prop:consupprop}
Let $\assemblygraph$ be a condensed supregraph and $v$ its vertex.
\begin{enumerate}
 \setlength{\itemsep}{0pt}   % no extra space between items
  \setlength{\parskip}{0pt}   % no paragraph spacing in items
  \setlength{\topsep}{0pt}    % no space above/below list
  \setlength{\parsep}{0pt}    % no extra spacing between paragraphs in item
  \setlength{\partopsep}{0pt} % no extra spacing when starting a new paragraph
  \item Either $v$ has at least two outgoing prefix edges and no outgoing suffix edges or exactly one outgoing suffix edge and no outgoing prefix edges.
  \item Either $v$ has at least two incoming suffix edges and no incoming prefix edges or exactly one incoming prefix edge and no incoming suffix edges.
  \item $v$ is unbranching (both indegree and outdegree equal to $1$ if only if it has a maximal label by substring inclusion among all vertex labels. Its incoming and outgoing edges are prefix and suffix respectively. We refer to such vertices as \emph{outer} vertices.
  \item $v$ has more than one incoming and outgoing vertices if only if it has a minimal label by substring inclusion among all vertex labels. Its incoming and outgoing edges are suffix and prefix respectively. We refer to such vertices as \emph{junction} vertices.
\end{enumerate}
\end{proposition}
\textbf{Proof: (1)}: We have already shown that condensed supregraphs are non-redundant, thus they can not have simultaneously outgoing prefix and suffix edges (see proposition~\ref{prop:speprop}). And any vertex can only have one outgoing suffix edge. It remains to prove that $v$ can not have exactly one outgoing prefix edge and no other outgoing edges. Consider a cycle $C$, such that $v$ is a substring of $Label(C)$. since the graph is non-redundant, $C$ must contain vertex $v$ at corresponding position. Since $v$ has only one outgoing edge $e$, $C$ must contain $e$ after this occurrence of $v$. Therefore, whenever $v$ is present in $Label(C)$, it can be extended to the occurrence of $e$ in $Label(C)$. This contradicts the assumption that $v$ is MES.

\textbf{(2)} The claim is symmetric to (1) and can be proven analogously.

\textbf{(3)} If no other vertex is a superstring of $v$, it can not have outgoing prefix or incoming suffix edges. Therefore by items (1) and (2) it must have exactly one incoming prefix edge and exactly one outgoing suffix edge. On the other hand, if a vertex has indegree and outdegree $1$, then it must have outgoing suffix ($uv$) and incoming prefix ($vw$) edges by items (1) and (2). Assume that $v$ is not a maximal substring vertex. Then the larger vertex must contain either $u$ or $w$ as a proper infix, which is impossible in a non-redundant graph.
% Add a specific lemma in non-reundancy section that can be used here.

\textbf{(4)} Opposite to item (3), vertex $v$ that has minimal label can only have outgoing prefix and incoming suffix edges. By items (1) and (2) there must be more than one such incoming and outgoing edges. On the other hand, if $v$ has outgoing prefix ($vw$) and incoming suffix ($uv$) edges, but is not a minimal vertex, then the vertex with label contained within $v$ must be a proper infix of either $u$ or $w$, which again is  impossible in a non-redundant graph.
\qed

Item (3) in proposition~\ref{prop:consupprop} shows that some of the vertices can be unbranching in condensed supregraph. However it is also evident from the proposition that two unbranching vertices can not be adjacent in the graph: only the maximal my substring inclusion vertices can be unbranching in condensed supregraph and such vertices are only connected with vertices that are their substrings, which can not therefore be unbranching. Thus, similarly to condensed de Bruijn graph, in condensed supregraph, unbranching vertices are merged together where possible.

Items (4) and (5) provide a classification of vertices in a condensed supregraph into junction, outer and all the remaining, which we will refer to as \emph{forks}. Every fork vertex has branching on one side (incoming or outgoing) and one incident edge on another side. Branching side always extends the sequence of fork vertex (prefix edges if outgoing, suffix edges if incoming), while the edge on the other side leads to a labeled with a prefix or suffix of the fork vertex. Any walk in a condensed supregraph can be viewed as a sequence of alternating junction and outer vertices, with fork vertices in between.

\begin{theorem}
For a finitely-conducting chromosome-set $\chromosomes$ $AG(\chromosomes, \Core(\chromosomes))$ is the smallest supregraph $\assemblygraph$ (in terms of number of vertices or edges) such that $\CPC(\assemblygraph)=\chromosomes$.
\end{theorem}

\textbf{Proof}. The proof is based on a stronger fact proven in theorem~\ref{thm:spgstar} item 3. Since we do not rely on the current theorem in the proof of $AG(\chromosomes, \Core(\chromosomes))$, it does not create a loop in the inference. Theorem~\ref{thm:spgstar} item 3 states that any graph $\assemblygraph^*$ can be obtained from $\SPG(\CPC(\assemblygraph*))$ using a sequence of free multiplexing operations. In other words. any graph $\assemblygraph^*$, such that $\CPC(\assemblygraph^*)=\chromosomes$ can be obtained from $\assemblygraph=\SPG(\chromosomes)=AG(\chromosomes, \Core(\chromosomes))$ using free multiplexing steps. Free multipexing only increases the number of vertices and edges in the graph. Condensing procedure, performed after each step in multiplexing process does not change the graph because free multiplexing never introduces unubranching paths in the graph. Since $\assemblygraph^*$ can be obtained from $\assemblygraph$ by modifications that increase the number of vertices and edges, $\assemblygraph$ is the minimal graph, such that $\CPC(\assemblygraph)=\chromosomes$.\qed

\begin{theorem}%\label{thm:cspg}
For any non-redundant assembly graph $\assemblygraph$, $\CSPG(\assemblygraph)$ is a condensed supregraph with the same set of chromosome path candidates: $\CPC(\assemblygraph) = \CPC(\CSPG(\assemblygraph))$.
\end{theorem}

\textbf{Proof}. We need to prove the following statements:
First we note that each operation in supregraph condensing procedure directly preserved the set of chromosome path candidates since labels of all circuits remained the same. Thus $\CPC(\assemblygraph) = \CPC(\CSPG(\assemblygraph))$.

To show that the resulting graph is supregraph we note that replacing edges with outer vertices ensures that every edge is either a prefixx or a suffix edge. The replacement of sequences with their MES extensions does not change that by lemma~\ref{lm:mes_substring}. Collapsing trivial paths also does not change that. Thus in $\CSPG(\assemblygraph)$ every edge is either a prefix or a suffix edge.

Finally, to show that the resulting graph is non-redundant Consider a cicruit $C$ in $\CSPG(\assemblygraph)$. For every vertex $v$ from $\assemblygraph$, its occurrences in $Label(C)$ match occurrences of $MES(v)$ in $Label(C)$ by definition of MES. On the other hand by construction every vertex in $\assemblygraph$ can be (surjectively) mapped to a vertex in $\CSPG(\assemblygraph)$, in such a way that every vertex occurrence is embedded into its MES occurrence. Thus the occurrences of MES extensions of every vertex match the occurrences of vertices in $\CSPG(\assemblygraph)$, showing that the graph is non-redundant.
% Incomplete! Need to introduce corresponding criterion for non-redundancy.
\qed

Note that every outer vertex $v$ has exactly one incoming ($uv$) and one outgoing ($vw$) edge. Thus it can be replaced with an edge $uw$, marked with the same sequence as $v$. This modification does not change the set of chromosome path candidates and it preserves non-redundancy. On the other hand this graph is smaller than $\SPG(\chromosomes)$. To truly find the minimal graph with a given set of chromosome path candidates one should apply this procedure to all outer edges in the graph.

\begin{proposition}\label{prop:min}
    Consider a finitely-conducting chromosome-set $\chromosomes$ and a graph $\assemblygraph$, constructed by replacing every outer vertex in graph $\SPG(\chromosomes)$ with an edge. Then $\assemblygraph$ is the smallest non-redundant graph in terms of the number of edges), such that $\CPC(\assemblygraph)=\chromosomes$.
\end{proposition}
\textbf{Proof}.
Consider an arbitrary non-redundant graph $\assemblygraph$ and apply the supregraph condensing procedure to it. When transforming every edge into a new vertex, the number of vertices and edges in the graph increases by the initial number of edges. However the last procedure that collapses trivial paths removes some of the edges. We will show that the actual increase in the number of vertices and edges equals the amount of outer edges in the graph $\CSPG(\assemblygraph)$.
Consider an edge $e$ from $u$ to $v$ in graph $\assemblygraph$ and assume that vertex $u$ has outdegree $1$. Since $\assemblygraph$ is non-redundant, every occurrence of $u$ in $\SPG(\assemblygraph)=\chromosomes$ must be an occurrence of vertex $u$ in a corresponding cycle. Every occurrence of $u$ in a cycle must continue with $e$ since $e$ is the only outgoing edge for $u$. Therefore $e$ is a substring of $MES(v)$ and thus $MES(e)=MES(u)$. Similarly if indegree of $v$ is $1$, we can conclude that $MES(e)=MES(v)$. Thus if either outdegree of $u$ or indegree of $v$ is $1$, after transforming every edge into vertex and then replacing vertex sequences with their MES extensions, a vertex, corresponding to edge $e$ will be adjacent to a vertex with exactly the same sequence and a trivial edge, connecting them. Such trivial edges are collapsed as the final step, reducing the number of vertices and edges in the graph by $1$. Thus in the end, the transformation of an edge into vertex adds extra vertex and edge to the final graph only if edge start has outdegree at least $2$ and edge end has indegree at least $2$, which corresponds to outer edges in $CSPG(\assemblygraph)$.

$\CPC(CSPG(\assemblygraph))=\chromosomes$ and it is a condensed supregraph. Thus it can be obtained from $\SPG(\chromosomes)$ using only free multiplexing steps. Let $E^*(\assemblygraph)$ be the number of edges in the graph minus the number of outer vertices in the graph. This value can only increase after free multiplexing operations, thus $E^*(\CSPG(\assemblygraph))\le E^*(\SPG(\chromosomes))$. On the other hand we have shown that $|E(\assemblygraph)|\ge E^*(\CSPG(\assemblygraph))$. Thus for any graph $\assemblygraph$, such that $\CPC(\assemblygraph)=\chromosomes$, we have $|E(\assemblygraph)|\ge E^*(\SPG(\chromosomes))$. On the other hand replacing all outer vertices with edges in $\SPG(\chromosomes)$ results in a graph with exactly $E^*(\SPG(\chromosomes))$ edges. Thus this graph has the minimal amount of edges among all graphs with a given set of chromosome path candidates.
\qed

\appsection{Mutiplexing}\label{app:multiplexing}
\begin{theorem}%\label{thm:mescnt}
For non-redundant assembly graph $\assemblygraph$, vertex $v$, and string $s$, if $Label(v)=MES(s)$ (with respect to $\CPC(\assemblygraph)$), then for any circuit $C$ in $\assemblygraph$, the number of occurrences of $v$ in $C$ is the same as the number of occurrences of $s$ in $Label(C)$.
\end{theorem}
\textbf{Proof}. Let $C$ be a circuit in $\assemblygraph$. By theorem~\ref{thm:mesext}, occurrences of $s$ in any chromosome candidate match with occurrences of $MES(s)$ in it. Since the graph is non-redundant, occurrences of $v=MES(s)$ in $Label(C)$ match the occurrences of vertex $v$ in circuit $C$. Thus, the number of occurrences of $s$ in $Label(C)$ is the same as the number of occurrences of $v$ in $C$.\qed

\begin{theorem}%\label{thm:freemultiplexing}
Let $\assemblygraph^*$ be obtained from a supregraph $\assemblygraph$ by multiplexing. Then:
\begin{enumerate}
 \setlength{\itemsep}{0pt}   % no extra space between items
  \setlength{\parskip}{0pt}   % no paragraph spacing in items
  \setlength{\topsep}{0pt}    % no space above/below list
  \setlength{\parsep}{0pt}    % no extra spacing between paragraphs in item
  \setlength{\partopsep}{0pt} % no extra spacing when starting a new paragraph
  \item $\assemblygraph^*$ is a condensed supregraph and $\CPC(\assemblygraph^*) \subseteq \CPC(\assemblygraph)$.
    \item $\CPC(\assemblygraph^*) = \CPC(\assemblygraph)$ if and only if every multiplexing step is free.
    \item Any condensed supregraph $\assemblygraph^*$ can be obtained from $\SPG(\CPC(\assemblygraph^*))$ (minimal supregraph with the same set of chromosome candidates) by a sequence of free multiplexing steps.
\end{enumerate}
\end{theorem}
\textbf{Proof: (1)} Assume that $\assemblygraph^*$ was obtained from $\assemblygraph$ in a single multiplexing step, but without the supregraph condensing procedure. Assume that the multiplexing was applied at vertex $v$ and resulted in a set of new vertices $U$. Consider a cycle $C$ in graph $\assemblygraph^*$. $C$ may pass through new vertices. Consider a sequence of vertices in $C$. If we replace each occurrence of of a vertex from $U$ with vertex $v$. The resulting sequence of vertices represents a cycle in graph $\assemblygraph$ that has the same label as $C$. Therefore label of any circuit in $\assemblygraph^*$ is aalso label of a circuit in $\assemblygraph$. By theorem~\ref{thm:cspg} supregraph compression does not change the set of circuit labels. Applying this observation to every step in the multiplexing process we conclude that $\CPC(\assemblygraph^*) \subseteq \CPC(\assemblygraph)$ regardless of the number of multiplexing steps it took to obtain $\assemblygraph^*$ from $\assemblygraph$.

\textbf{(2)} Again assume that $\assemblygraph^*$ was obtained from $\assemblygraph$ in a single multiplexing step, but without the supregraph condensing procedure. First, consider the case that the multiplexing step was free and every pair of incoming and outgoing edges of vertex $v$ became a new vertex in $\assemblygraph^*$. Any circuit in graph $\assemblygraph$ can be transformed into a circuit in $\assemblygraph^*$ by replacing all occurrences of $v$ with one of the new vertices, chosen based on the edges before and after $v$ in the cycle. Since supregraph condensing also does not change the set of circuit labels, free multiplexing does not change it too.

Assume now that the multiplexing step was not free and there is a pair of edges $uv,vw$ that are not represented in $\assemblygraph^*$ as a new vertex. Consider also a graph $\assemblygraph^f$, that is the result of free multiplexing of graph $\assemblygraph$ in vertex $v$. In $\assemblygraph^f$ a combined sequence of vertices $u$ and $w$ is present in at least one cycle $C^f$. This cycle must contain the new vertex $v'$, representing the combined sequence of $u$ and $w$. Assume that $\CPC(\assemblygraph^*)=\CPC(\assemblygraph)=\CPC(\assemblygraph^f)$. In such case $Label(C)$ must be a label of a cycle in $\assemblygraph^*$. This cycle $C^*$ can not contain $v'$, since $v'$ was not added to $\assemblygraph^*$ during multiplexing step. Since $\assemblygraph^*$ is a subgraph of $\assemblygraph^f$, $C^*$ is also a cycle in $\assemblygraph^f$. Thus two different cycles in $\assemblygraph^f$ have the same label, which is impossible in a non-redundant graph, a contradiction. Therefore the set of chromosome path candidates only remains unchanged after free multiplexing steps.\qed

Before we proceed to item (3) we will prove the following lemmas.

\begin{lemma}\label{lm:supremes}
Let $\assemblygraph$ be a compressed supregraph. Consider vertex $u$ that has outgoing prefix edges. $\assemblygraph$ contains edge $uv$ iff $Label(v)$ is a minimal by substring inclusion MES with respect to $\CPC(\assemblygraph)$ that contains $u$ as a proper prefix.
\end{lemma}
\textbf{Proof}. Assume $s$ is a MES, that contains $u$ as a proper prefix, but itself is a proper prefix of $Label(v)=Label(uv)$. From theorem~\ref{thm:mesalt} there exists letter $c$ such that $sc$ is a substring of $\CPC(\assemblygraph)$, but not a prefix of $v$. Consider cycle $C$, such that $Label(C)$ that contains $sc$. Since $\assemblygraph$ is non-redundant, $C$ must contain vertex $v$ at the start of occurrence of $sc$ in $Label(C)$. It must be followed by edge $uv$ since the first letter after $u$ in $v$ is the same as in $s$ and $\assemblygraph$ is splitting. But $sc$ is not a prefix of $v$, a contradiction. Thus every outgoing edge from $u$ is a minimal MES, that contains $u$ as a proper prefix.

To prove the inverse statement consider a minimal MES $s$ that contains $u$ as a proper prefix. There must be a cycle $C$ in the graph, such that $s$ is a substring of $Label(C)$. This substring occurrence must include occurrence of $u$. The next edge after $u$ in the cycle must be a prefix edge $uv$. $v$ is a MES, that contains $u$ as a proper infix. Thus $s$ is a prefix of $v$. On the other hand $v$ is also a minimal MES, that contains $u$ as a prefix. Therefore $s=v$.
\qed

\begin{lemma}\label{lm:spgpo}
    Consider the following relation on supregraphs: for assembly graphs $\assemblygraph_1, \assemblygraph_2$ we say that $\assemblygraph_1\le\assemblygraph_2$ if every vertex in $\assemblygraph_2$ is a superstring of a vertex in $\assemblygraph_1$. When restricted to the set of all condensed supregraphs with a given set of chromosome candidates $\chromosomes$, this relation is a partial order relation.
\end{lemma}
\textbf{Proof}. Reflexivity and transitivity of this relation are trivial. We will prove antisymmetry: whenever $\assemblygraph_1\le\assemblygraph_2$ and $\assemblygraph_2\le\assemblygraph_1$, we have $\assemblygraph_1=\assemblygraph_2$.

Since $\assemblygraph_1$ and $\assemblygraph_2$ are non-redundant it is sufficient to prove that they have the same set of vertices. Indeed, in such case by item 1 of theorem~\ref{thm:ag} $$\assemblygraph_1=\AG(\CPC(\assemblygraph_1), V(\assemblygraph_1))=\AG(\chromosomes, V(\assemblygraph_1))=V(\assemblygraph_1))=\AG(\CPC(\assemblygraph_2), V(\assemblygraph_2))=\assemblygraph_2.$$

Consider the smallest vertex in the graphs $\assemblygraph_1$ and $\assemblygraph$ that does not have a copy in the other graph. Without loss of generality, assume that it is vertex $v_1$ from graph $\assemblygraph_1$. Since $\assemblygraph_2\le\assemblygraph_1$, there is a vertex in $\assemblygraph_2$, that is a substring of $v$. Consider the largest such vertex $u_2$. $u_2$ can not have the same label as $v$, thus it is shorter than $v$ and must have a copy $u_1$ in $\assemblygraph_1$. By item (4) of proposition~\ref{prop:speprop}, $u_1$ cannot be a proper infix of $v_1$. Without loss of generality, assume that $u_1$ is a prefix of $v_1$. By item (2) of proposition~\ref{prop:speprop} there must be a path from $u_1$ to $v_1$, that consists of only prefix edges. This path cannot contain more than one edge since otherwise there will be a vertex, larger than $u_2$ in $\assemblygraph_2$, that is a substring of $v_1$. Thus, there exists a prefix edge $u_1v_1$. Consider minimal MES $s$, that contains $u_1$ as a proper prefix and is itself a prefix of $v_1$. It exists since $v_1$ itself is MES. By lemma~\ref{lm:supremes} both $\assemblygraph_1$ and $\assemblygraph_2$ must contain vertex with label $s$ and there should be an edge from $u_1$/$u_2$ to this vertex. Since edge $u_1v_1$ is the only edge in graph $\assemblygraph$ that can have $s$ as a label, we have $v_1=s$. But $\assemblygraph_2$ must contain a vertex marked with $s$ too, thus there exists vertex $v_2$ in $\assemblygraph_2$ that has the same label as $v_1$.\qed

We are now ready to proceed with the proof of item (3) of theorem~\ref{thm:freemultiplexing}.

\textbf{(3)} Let $\chromosomes=\CPC(\assemblygraph^*)$ and $\assemblygraph=\SPG(\chromosomes)$. To obtain $\assemblygraph^*$ from $\assemblygraph$ we will iteratively apply free multiplexing to a smallest vertex in $\assemblygraph$ that is not present in $V(\assemblygraph^*)$.

First we will show that such operations retain invariant $\assemblygraph\le\assemblygraph^*$. This property holds for initial $\assemblygraph$, that is $AG(\chromosomes, \Core(\chromosomes))$ because $\Core(\chromosomes))$ contains all conductor MES, that are not proper infixes of the conductor MES. Thus every conductor MES (and thus every vertex of $\assemblygraph^*$) contains a string from $\Core(\chromosomes)$.

Free operation on a junction vertex also does not change the invariant. Indeed, consider a vertex $v$ from graph $\assemblygraph$ that is not present in $V(\assemblygraph^*)$. Consider vertex $u$ in $\assemblygraph^*$, that contains $v$ as a substring: $u=avb$. Without loss of generality assume that $b$ is not empty. Let $s$ be the shortest MES substring of $u$ that contains $v$ as a proper prefix. It exists since $vb$ is a MES.
% Add this property to conductor MES discussion
By lemma~\ref{lm:supremes} $s$ is a vertex in $\assemblygraph$ and $vs$ is an edge it. $s$ will stay in the graph after the free multiplexing procedure since $v$ is the only vertex that is removed. Thus any vertex in $\assemblygraph^*$ that had a substring vertex in $\assemblygraph$ still has one after multiplexing operation on a junction vertex that is not present in $\assemblygraph^*$.

We still need to show that it is possible to apply multiplexing to the shortest vertex in $\assemblygraph$, that is not present in $\assemblygraph^*$, that is to show that such vertex is always a junction. Assume that the smallest such vertex $v$ is not junction. Without loss of generality assume that it has incoming prefix edge $uv$. $u$ is shorter than $v$, thus there is vertex $u^*$ in $\assemblygraph^*$. By lemma~\ref{lm:supremes} $v$ is a smallest MES that contains $u$ as a proper prefix. Thus vertex with label $v$ must be present in $\assemblygraph^*$ too, a contradiction.

% Explain why process is not infinite
The process stops when $V(\assemblygraph)\subseteq V(\assemblygraph^*)$. Thus $\assemblygraph^*\le \assemblygraph$. However we have also shown that $\assemblygraph\le \assemblygraph^*$ at any moment in the multiplexing process. Thus by lemma~\ref{lm:spgpo} $\assemblygraph=\assemblygraph^*$ in the end of multiplexing process.\qed

We previously noted that one can reinterpret reads as sequence of $k$-mers by replacing the original alphabet with the alphabet of $k$-mers. However, in the formulation of next theorem implicitly relies on $k$ when defining $\SPG_0(\reads)$ as $\CSPG(\DBG(\reads,k))$. This choice is made to keep the construction of $\SPG(\reads)$ as close to practical implementations as possible. If we instead work strictly over the $k$-mer alphabet, the appropriate definition would be $\SPG_0(\reads) = \CSPG(\DBG(\reads,1))$, where vertices correspond to single letters of the $k$-mer alphabet and edges correspond to adjacent pairs of such letters. We adopt this interpretation in the proof of the next theorem.

\begin{theorem}
    Multiplexing with decision rule $D^{SPG}$ always stops in a finite number of steps and regardless of operation order $\ExhaustiveMultiplex(SPG_0(\reads), D^{SPG})=SPG(\reads)$.
\end{theorem}
\textbf{Proof}.
Let $\assemblygraph$ be any graph, that can be obtained from $SPG_0$ through multilexing with the decision rule $D^{SPG}$. We will show that $\assemblygraph$ always retains the following properties: $\assemblygraph\le \SPG(\reads)$, $\CC(\reads)\subseteq\CPC(\assemblygraph)$, and every edge in $\assemblygraph$ is a substring of $\CPC(\reads)$.

First we show that $\SPG_0(\reads)=\CSPG(\DBG(\reads,k))$ satisfies these conditions. Combining theorems \ref{thm:conventional} and \ref{thm:cspg} we have $\CC(\reads)\subseteq\CPC(\DBG(\reads,k))=\CPC(SPG_0(\reads))$. Next we note that every edge label in de Bruijn graph is a substring of $\CPC(\reads)$ by construction. Replacing edge and vertex sequences with their MES extensions during supregraph compressing of the de Bruijn graph does not change this fact since the cyclic string from $\CPC(\reads)$ that contained the sequence of a given edge as substring, also must contain its MES extension as a substring. Finally, to show that $\SPG_0(\reads)\le \SPG(\reads)$ consider a vertex $v$ in $\SPG(\reads)$. $v$ is a MES with respect to $\CPC(\reads)$, thus $v$ is also a MES with respect to a larger chromosome-set $\CPC(\SPG_0(\reads))$. Consider a single letter $c$ of $v$. $c$ corresponds to a vertex in $\DBG(\reads, 1)$ supregraph compressing replaces $c$ with its MES extension. Since MES extension preserves substring relation, MES extension of $c$ is a substring of MES extension of $v$, which is equal to $v$. Therefore for any vertex $v$ in $\SPG(\reads)$ there is a vertex in $\SPG_0(\reads)$, labeled with a substring of $v$, thus $\SPG_0(\reads)\le \SPG(\reads)$.

Next we show that multiplexing preserves these properties of the graph. Assume that the multiplexing was applied to a vertex $v$ in supregraph $\assemblygraph$, resulting in a supregraph $\assemblygraph'$. Consider a circuit $C$ in $\assemblygraph$, such that $Label(C)\in\CC(\reads)$. If $C$ does not pass through $v$, then it remains in $\assemblygraph'$. If $C$ passes through $v$, the combined label of edges before and after $v$ in $C$ is a substring of $\label(C)$, and thus a substring of $\CC(\reads)$. Thus this pair was retained during multiplexing. Replacing all occurrences of $v$ in $C$ with corresponding vertex, created during multiplexing results in a cycle in $\assemblygraph'$ with the same label. Therefore labels of all cycles in $\assemblygraph$ that are also in $\CC(\reads)$ are preserved during multiplexing. Thus $\CC(\reads)\subseteq\CPC(\assemblygraph')$.

New vertices and edges added during the multiplexing step are specifically checked by the $D^{SPG}$ condition to be substrings of $\CC(\reads)$. Thus all vertex and edge sequences in $\assemblygraph'$ are still substrings of $\CC(\reads)$.

Finally, consider a vertex $u=v$ in $\SPG(\reads)$ that contains $v$ as a substring. Consider a smallest MES with respect to $\CPC(\assemblygraph)$, that contains $v$ as a proper substring and is a substring of $u$. It exists since $u$ itself is a MES with respect to $\CC(\reads)$ and thus it is a MES $s$ with respect to a larger set $\CPC(\assemblygraph)$. By lemma~\ref{lm:supremes} $s$ is a vertex adjacent to $v$ in $\assemblygraph$. Since $v$ is the only deleted vertex, we infer that any vertex sequence in $\SPG(\reads)$ has at least one vertex from $\assemblygraph'$ as a substring. Therefore $\assemblygraph'\le\SPG(\reads)$.

This concludes the proof that the multiplexing process preserves three properties of the graph: $\assemblygraph\le \SPG(\reads)$, $\CC(\reads)\subseteq\CPC(\assemblygraph)$, and every edge in $\assemblygraph$ is a substring of $\CPC(\reads)$. 

To complete the proof assume that $\assemblygraph$ is the final result of the multiplexing process. The process stops when no new junction can be selected for multiplexing. Thus every junction vertex is a conductor with respect to $\CC(\reads)$. Since every vertex contains junction vertex as a substring (minimal vertex contained within a given vertex as a substring is always junction), every vertex in the graph is a conductor. We can show by induction on the walk length that every walk label in $\assemblygraph$ is a substring og $\SPG(\reads)$. As a base we have shown that every edge in the graph is a substring of $\SPG(\reads)$. Consider a walk $P$ and an edge $uv$, where $u$ is the last vertex in $P$. By induction assumption, $Label(P)$ and $Label(uv)$ are substrings of $\SPG(\reads)$. Since $u$ is a conductor, by lemma~\ref{lm:switch}, we infer that $Label(P+uv)$ is also a substring of $\CC(\reads)$. Consider an arbitrary circuit $C$ in $\assemblygraph$ and a walk $P$, obtained from $C$ by breaking it in vertex $v$. $Label(P)$ is a substring of $\CC(\reads)$ and $v$ is a conductor. Thus by lemma~\ref{lm:circularization} $Label(C)\in\CC(\reads)$. Therefore $\CPC(\assemblygraph)\subseteq\CC(\reads)$. We have already shown that $\CC(\reads)\subseteq\CPC(\assemblygraph)$, therefore $\CC(\reads)=\CPC(\assemblygraph)$. By theorem~\ref{thm:freemultiplexing} $\SPG(\reads)\le\assemblygraph$. We have already shown that $\assemblygraph\le \SPG(\reads)$. Thus $\assemblygraph = \SPG(\reads)$, which shows that regardless of the order of steps, final result of multiplexing procedure is always $\SPG(\reads)$.

Finally to show that the process ends in a finite number of steps we note that every multiplexing step removes from the graph a non-conductor vertex. The same vertex can not be removed twice and there is a finite amount of non-conductor strings since $\CC(\reads)$ is finitely-conducting chromosome-set (theorem~\ref{thm:readscond}). Therefore the process is finite.
% Deal with cyclic vertices
\qed

\begin{theorem}%\label{thm:spgstar}
    Multiplexing with the decision rule $D^{\reads}$ always stops in a finite number of steps and regardless of the operation order the result is the same condensed supregraph $SPG^*(\reads)=\ExhaustiveMultiplex(SPG(\reads), D^{\reads})$. $SPG^*(\reads)$ contains vertex with label $MES(r)$ for every non-contained read $r$ from $\reads$.
\end{theorem}
\textbf{Proof}.
Since only proper infixes of reads are used for multiplexing operations and the same vertex can not be used twice, the length of the process is finite. Moreover, the length of the process is always the same since every conductor MES that is a proper infix of a read will be the focus of multiplexing operation exactly once during the process. Consider the final graph $SPG^*(\reads)$. We will show that the set of vertices in it is the same regardless of the order of operations.

Consider a vertex $v$ in this graph. It cannot be a proper infix of a read since otherwise the process would not stop at this graph. If this vertex was present in the initial graph, it would never be removed from the graph regardless of the order of operations. Thus this vertex belongs to any other possible final graph. If $v$ was not present in the initial graph, it was added during one of the multiplexing operations. Let multiplexing in vertex $u$ be the the operation that added $v$. At the moment of this operation, $u$ is a junction vertex. Thus by lemma~\ref{lm:supremes} vertices adjacent to $u$ at that moment are the minimal MES that contain $u$ as a proper substring. This set does not depend on the order of operations, thus vertex $v$ would be added to the graph regardless of the order of operations. $v$ also would never be removed since it is only removed when multiplexing operation is centered on it and it is impossible since $v$ is not a proper infix of any read. Therefore, regardless of the order of operation, the set of vertices (and thus the entire graph) will be the same.

Let $r$ be a non-contained read. Assume that $s$ is no proper infix of $MES(r)=asb$ that is a conductor MES but is not a proper infix of $r$. In such case $r$ is a substring of either $as$ or $sb$ (or both). Both $as$ and $sb$ and MES, a contradiction, since $MES(r)$ is the smallest MES that contains $r$ as a substring.

Assume $r$ does not contain a MES conductor as a proper infix. Then $MES(r)$ does not contain one also. By theorem~\ref{thm:readscond}, thus $MES(r)$ is a MES conductor. Thus $MES(r)\in\Core(\CC(\reads))=V(\SPG(\reads))$. In addition, $MES(r)$ cannot be removed from graph by multiplexing process since a superstring of a non-contained read cannot be a proper infix of another read. Thus $MES(r)$ is present in $\SPG^*(\reads)$.

Let $s$ be the largest conductor MES that is present a proper infix in $r$. Multiplexing must be applied to $s$ once.  Let $MES(r)=asb$. Consider the smallest $b'$ such that $sb'$ is a conductor MES and $b'$ is a prefix of $b$. If $b\neq b'$, $r$ either contains $sb'$ as a proper infix, which is impossible since $s$ is the largest such string, or $r$ is contained within $asb'$, which is impossible since $asb$ is the smallest MES, that contains $r$. Thus $b=b'$ and by lemma~\ref{lm:supremes}, $sb$ is a vertex that was adjacent to $s$ when multiplexing was applied to $s$. Similarly $as$ is also a vertex adjacent to $s$. Free multiplexing at $s$ adds $asb=MES(r)$ as a new vertex. As discussed above this vertex can never be removed from the graph, thus $MES(r)$ is present in $\SPG^*(\reads)$ in case $r$ contains a proper infix conductor MES too.
\qed

\begin{proposition}\label{prop:example}
    Let $\assemblygraph$ be a non-redundant splitting assembly graph without prefix and suffix edges, such that $\CPC(\assemblygraph)$ consists of all cyclic strings. If $\assemblygraph$ contains an edge of length at least $k+1$, it can not contain vertices shorter than $k$.
\end{proposition}
\textbf{Proof}.
We will show by induction on $1\le t\le k-1$ that there are no vertices of length $t$ in $\assemblygraph$.
Let $t$-mers $s_1$ and $s_2$ be neighbours ($s_1a=bs_2$ for some letters $a, b$). We will show that if $s_1$ is not a vertex in $\assemblygraph$, then $s_2$ also is not a vertex.

Assume that $\assemblygraph$ contains vertex $s_2$
Let $C$ be a cycle in the graph such that $Label(C)$ contains $cs_1b$ as a substring where $c$ is a letter. By induction assumption all vertices  have length at least $t=|s_1|$. Thus every $(t+1)$-mer from $Label(C)$ is a substring of at least one edge. If edges, covering $(t+1)$-mers $cs_1$ and $s_1b$ are different the vertex between them would be a substring of $s_1$, which is impossible. Thus for any letter $c$ there is edge $e_c$ that contains $cs_1b$ as a substring. Moreover, all these edges must end in vertex $s_2$ since is a suffix fo $cs_1b=cas_2$. However, all such edges end with $as_2$, which is impossible in a splitting graph, a contradiction. Therefore $s_2$ is not present a a vertex in $\assemblygraph$.

To complete the induction step we note that since $\assemblygraph$ contains an edge of length $k+1$, it can not contain vertices labeled as proper infixes of the edge. Thus for every $1\le t\le k-1$ there is at least one $t$-mer missing from the graph. Since every $t$-mer is reachable from any other $t$-mer we can conclude that no $t$-mer is present in $\assemblygraph$ for $1\le t\le k-1$.
\qed

Consider the example discussed at the end of section~\ref{sec:multiplexing}. There, we analyzed a read set capable of representing an arbitrary sequence over the alphabet. Constructing a genome candidate that preserves the representation of this read set—when one of the reads has length $k+1$—requires introducing a vertex or edge whose multiplicity matches that of this sequence. In particular, such an element must have length at least $k+1$. However, proposition~\ref{prop:example} implies that graphs of this type cannot contain vertices of length less than $k$. This establishes a lower bound on the structural complexity of the graph.

\appsection{Optimal genome assembly}\label{app:optimal}

\begin{theorem}
For any read set $\reads$ and a set of multiplicity conditions $\multiplicity$ defined on non-contained reads:
$$OptimalAssembly(\reads|\multiplicity)=\Superstrings(OA(GC(\reads|\multiplicity)), \reads^1),$$
where $\reads^1$ denotes the set of unique reads from $\reads$.
\end{theorem}

\textbf{Proof}. The condensing step merges each unbranching path into a single vertex. This vertex inherits flow conditions from all merged vertices. Therefore, the final condensed graph can also be considered as a network with capacities inherited from the network defined by $SPG^*(\reads)$ and $\vertexmultiplicity$. If any of the vertices in the unbranching path was unique, then the condensed vertex becomes the smallest one that contains the sequence of a unique read.

In absence of topology condition, any circuit decomposition of a circulation in the constructed network is a genome candidate. Thus we need to prove that the vertices marked as unique in the network are the longest sequences that are present in any decomposition of any valid circulation, assuming that every edge is present in at least one valid circulation (since we removed every edge that was not).

Consider a unique vertex $v$ and the shortest closed path $P$ that contains $v$. Every genome path candidate going through $v$ also goes through $P$. Also $Label(P)=Label(v)$. Let $P$ end in vertex $u$ and $e$ be an edge, outgoing from $u$. There is at least one circulation $C$ that uses edge $e$. We claim that there exists a decomposition of this circulation, that contains path $Pe$. Indeed $P$ and $e$ must each be present in any decomposition of $C$. If $P$ does not continue with $e$, it must continue with some edge $e'$, while $e$ is preceded by some walk $P'$. Switching in $u$ connections $Pe'$ and $P'e$ into connections $Pe$ and $P'e'$ creates a different decomposition but does not change vertex multiplicities. Thus we can find a decomposition that contains $Pe$ as a subpath. Since $v$ must be present exactly once in any path decomposition, there is no other path going through $v$ in this genome candidate, except the one containing $Pe$. Similarly one can find a genome path candidate that contains $Pe'$ but not $Pe$ for another edge $e'$, outgoing from $v$. Combining this with similar lggic applied to the start of path $P$, we conclude that it is impossible to prolong path $P$ to a longer path that is present in any circulation decomposition. Therefore the procedure UniqueOptimal constructs all optimal contigs that contain unique vertices as substrings.\qed

\end{document}